\documentclass[reqno]{amsart}%
\usepackage{amsfonts}
\usepackage{amsmath}
\usepackage{amssymb}
\usepackage{graphicx}%
\newtheorem{theorem}{Theorem}
\theoremstyle{plain}

\newtheorem{lemma}{Lemma}
\newtheorem{notation}{Notation}
\newtheorem{problem}{Problem}

\newtheorem{remark}{Remark}

\numberwithin{equation}{section}
\ifx\pdfoutput\relax\let\pdfoutput=\undefined\fi
\newcount\msipdfoutput
\ifx\pdfoutput\undefined\else
\ifcase\pdfoutput\else
\ifx\paperwidth\undefined\else
\ifdim\paperheight=0pt\relax\else\pdfpageheight\paperheight\fi
\ifdim\paperwidth=0pt\relax\else\pdfpagewidth\paperwidth\fi
\fi\fi\fi
\begin{document}
\title[$p$-Adic Statistical Field Theory and DBNs]{$p$-Adic Statistical Field Theory and Deep Belief Networks}
\author[Z\'{u}\~{n}iga-Galindo]{W. A. Z\'{u}\~{n}iga-Galindo}
\address{University of Texas Rio Grande Valley\\
School of Mathematical \& Statistical Sciences\\
One West University Blvd\\
Brownsville, TX 78520, United States }
\email{wilson.zunigagalindo@utrgv.edu}
\thanks{The author was partially supported by the Lokenath Debnath Endowed Professorship.}

\begin{abstract}
In this work we initiate the study of the correspondence between $p$-adic
statistical field theories (SFTs) and neural networks (NNs). In general
quantum field theories over a $p$-adic spacetime can be formulated in a
rigorous way. Nowadays these theories are considered just mathematical toy
models for understanding the problems of the true theories. In this work we
show these theories are deeply connected with the deep belief networks (DBNs).
Hinton et al. constructed DBNs by stacking several restricted Boltzmann
machines (RBMs). The purpose of this construction is to obtain a network with
a hierarchical structure (a deep learning architecture). An RBM
corresponds\ to a certain spin glass, we argue that a DBN should correspond to
an ultrametric spin glass. A model of such a system can be easily constructed
by using $p$-adic numbers. In our approach, a $p$-adic SFT corresponds to a
$p$-adic continuous DBN, and a discretization of this theory corresponds to a
$p$-adic discrete DBN. We show that these last machines are universal
approximators. In the $p$-adic framework, the correspondence between SFTs and
NNs is not fully developed. We point out several open problems.

\end{abstract}
\maketitle
\tableofcontents

\section{Introduction}

Recently, it has been proposed the existence of a correspondence between
neural networks (NNs) and quantum field theories (QFTs), more precisely with
Euclidean QFTs, see, e.g., \cite{Batchits et al}, \cite{Dyer et al},
\cite{Erbin}, \cite{Halverson et al}, \cite{Maiti}, \cite{Moritz Dahmen},
\cite{Roberts-Yaida}, \cite{Yaida}, see also \cite{Buice}, \cite{Buice2},
\cite{Chow}, and the references therein. This correspondence take different
forms depending on the architecture of the networks involved. This article
aims to initiate the study of the mentioned correspondence in the framework of
the non-Archimedean statistical field theory (SFT), see, e.g.,
\cite{Zunifa-RMP-2022}, \cite{Arroyo-Zuniga}, see also \cite{Abdesselam},
\cite{Kh1}-\cite{Kh2}, \cite{Kochubei et al}, \cite{LM89},
\cite{Mendoza-Zuniga}, \cite{Mis-1}-\cite{Mis-3}, \cite{Zuniga-JFAA}, and the
references therein. In \ this case, the corresponding NNs are new hierarchical
generalizations of the classical restricted Boltzmann machines (RBMs), see,
e.g., \cite{Decelle et al}, \cite{Fischer-Igel}. Here, we focus on the
$p$-adic counterparts of the convolutional deep belief networks (DBNs), see,
e.g., \cite{Dong et al}, \cite{Hinton et al}, \cite{Honglak et al}, and the
references therein.

A fundamental problem is the understanding of the structure of space-time at
the level of the Planck scale. In the 1930s Bronstein showed that general
relativity and quantum mechanics imply that the uncertainty $\Delta x$ of any
length measurement satisfies $\Delta x\geq L_{\text{Planck}}:=\sqrt
{\frac{\hbar G}{c^{3}}}$, where $L_{\text{Planck}}$ is the Planck length
($L_{\text{Planck}}\approx10^{-33}$ $cm$). This inequality implies that
space-time is not an infinitely divisible continuum (mathematically speaking,
the spacetime must be a completely disconnected topological space at the level
of the Planck scale). Bronstein's \ inequality has motivated the development
of several different physical theories. At any rate, this inequality implies
the need of using non-Archimedean mathematics in models dealing with the
Planck scale. In the 1980s, Volovich proposed the conjecture that the
space-time at the Planck scale has a $p$-adic nature, see, e.g.,
\cite{Volovich1}. This conjecture has propelled a wide variety of
investigations in cosmology, quantum mechanics, string theory, QFT, etc., and
the influence of this conjecture is still relevant nowadays, see, e.g.,
\cite{Abdesselam}, \cite{Arroyo-Zuniga}-\cite{Av-5}, \cite{Bocardo-Zuniga-1}%
-\cite{B-F}, \cite{Dra01}-\cite{DD97}, \cite{Fuquen et al}-\cite{GC-Zuniga},
\cite{Gubser}-\cite{Gubser et al-1}, \cite{Harlow et al}, \cite{Kh1}%
-\cite{Khrennikov-Kozyrev}, \cite{Koch}-\cite{Kochubei et al}, \cite{LM89},
\cite{Mendoza-Zuniga}, \cite{Mis-1}-\cite{M-P-V}, \cite{Mukhamedov-1}%
-\cite{R-T-V}, \cite{V-V-Z}-\cite{Volovich1}, \cite{Zabrodin}%
-\cite{Zuniga-JFAA}.

A $p$-adic number is a series of the form%
\begin{equation}
x=x_{-k}p^{-k}+x_{-k+1}p^{-k+1}+\ldots+x_{0}+x_{1}p+\ldots,\text{ with }%
x_{-k}\neq0\text{,} \label{p-adic-number}%
\end{equation}
where $p$ is a fixed prime number, and the $x_{j}$s \ are numbers in the set
$\left\{  0,1,\ldots,p-1\right\}  $. The set of all possible series of the
form (\ref{p-adic-number}) constitutes the field of $p$-adic numbers
$\mathbb{Q}_{p}$. There are natural field operations, sum and multiplication,
on series of the form (\ref{p-adic-number}), see, e.g., \cite{Koblitz}. There
is also a natural norm in $\mathbb{Q}_{p}$ defined as $\left\vert x\right\vert
_{p}=p^{k}$, for a nonzero $p$-adic number of the form (\ref{p-adic-number}).
The field of $p$-adic numbers with the distance induced by $\left\vert
\cdot\right\vert _{p}$ is a complete ultrametric space. The ultrametric (or
non-Archimedean) property refers to the fact that $\left\vert x-y\right\vert
_{p}\leq\max\left\{  \left\vert x-z\right\vert _{p},\left\vert z-y\right\vert
_{p}\right\}  $ for any $x$, $y$, $z\in\mathbb{Q}_{p}$. We denote by
$\mathbb{Z}_{p}$ the unit ball, which consists of all series with expansions
of the form (\ref{p-adic-number}) with $-k\geq0$. The unit ball is an infinite
rooted tree, with valence $p$. The field $\mathbb{Q}_{p}$ has a tree-like
(hierarchical) structure. We extend the $p-$adic norm to $\mathbb{Q}_{p}^{N}$
by taking $||x||_{p}=\max_{1\leq i\leq N}|x_{i}|_{p}$, for $x=(x_{1}%
,\dots,x_{N})\in\mathbb{Q}_{p}^{N}$.

The space $\mathbb{Q}_{p}^{N}$ has a very rich mathematical structure. The
axiomatic quantum field \ theory can be extended to $\mathbb{Q}_{p}^{N}$, see,
e.g., \cite{Glimm-Jaffe}, \cite{Kleinert et al}, \cite{Simon-0}\ for the
classical theory. In \cite{Mendoza-Zuniga}, a family of quantum scalar fields
over a $p-$adic spacetime which satisfy $p-$adic analogues of the
G\aa rding--Wightman axioms was constructed. In \cite{Arroyo-Zuniga}, a large
class of interacting Euclidean quantum field theories was constructed by using
white noise calculus. These quantum fields fulfill all the
Osterwalder-Schrader axioms, except the reflection positivity. In
\cite{Zunifa-RMP-2022}, the author constructs, in a rigorous mathematical way,
interacting Euclidean quantum field theories on a $p$-adic spacetime. The main
result is the construction of a measure on a function space which allows a
rigorous definition of the partition function. The advantage of the approach
presented is that all the perturbation calculations can be carried out in the
standard way using functional derivatives, but in a mathematically rigorous
way. In \cite{Abdesselam} Abdesselam et al. present the construction of scale
invariant non-Gaussian generalized stochastic processes over three dimensional
$p$-adic space. The construction includes that of the associated squared
field, this field has a dynamically generated anomalous dimension which
rigorously confirms a prediction made more than forty years ago by K. G.
Wilson. Traditionally the $p$-adic QFTs has been considered just mathematical
toy models. In this article, we show that these theories are deeply connected
with hierarchical versions of RBMs, and then with deep learning.

An Euclidean quantum field theory is a probability measure of the form%
\[
d\mathbb{P}\left(  \varphi\right)  =\frac{e^{-E\left(  \varphi\right)
}d\mathbb{P}_{0}\left(  \varphi\right)  }{\int_{H}e^{-E\left(  \varphi\right)
}d\mathbb{P}_{0}\left(  \varphi\right)  }%
\]
on a space $H$ of functions $\varphi:\mathbb{Q}_{p}^{N}\rightarrow\mathbb{R}$,
where $\mathbb{P}_{0}$\ is a Gaussian measure on $H$. For the sake of
simplicity, along this article we assume that $N=1$. By a discretization
process, which consists in finding the restriction of $\mathbb{P}$ to a
suitable finite dimensional vector subspace $H_{l}$ of $H$, one obtains a
discrete energy functional $E_{l}$ and a finite dimensional Boltzmann
distribution $\mathbb{P}_{l}$, for $l\geq l_{0}$, such that $\mathbb{P}%
_{l}\rightarrow\mathbb{P}$ in some sense, see \cite{Zunifa-RMP-2022}\ and the
references therein. The discrete energy functional of a $\phi^{4}$-theory has
the form%
\[
E_{l}\left(  \phi\right)  =%
%TCIMACRO{\dsum \limits_{i,j\in G_{l}}}%
%BeginExpansion
{\displaystyle\sum\limits_{i,j\in G_{l}}}
%EndExpansion
\phi_{i}w_{i,j}^{l}\phi_{j}+%
%TCIMACRO{\dsum \limits_{i\in G_{l}}}%
%BeginExpansion
{\displaystyle\sum\limits_{i\in G_{l}}}
%EndExpansion
a_{i}^{l}\phi_{i}+%
%TCIMACRO{\dsum \limits_{i\in G_{l}}}%
%BeginExpansion
{\displaystyle\sum\limits_{i\in G_{l}}}
%EndExpansion
b_{i}^{l}\phi_{i}^{2}+%
%TCIMACRO{\dsum \limits_{i\in G_{l}}}%
%BeginExpansion
{\displaystyle\sum\limits_{i\in G_{l}}}
%EndExpansion
b_{i}^{l}\phi_{i}^{4},
\]
where $G_{l}=\mathbb{Z}_{p}/p^{l}\mathbb{Z}_{p}\simeq\mathbb{Z}/p^{l}%
\mathbb{Z}$ is the additive group of the integers modulo $p^{l}$, and
$\phi=\left[  \phi_{i}\right]  _{i\in G_{l}}$.

We identify the elements of $G_{l}$ with integers of the form $i\boldsymbol{=}%
i_{0}+i_{1}p+\ldots+i_{l-1}p^{l-1}$. The restriction of $\left\vert
\cdot\right\vert _{p}$ to $G_{l}$ induces a norm, and thus $G_{l}$ is a finite
ultrametric space. In addition, $G_{l}$ can be identified with the set of
branches (vertices at the top level) of a rooted tree with $l+1$ levels (or
layers) and $p^{l}$ branches.

\textit{ A }$p$\textit{-adic discrete deep belief network} is a discrete
Euclidean QFT defined by an energy functional of the form%
\[
E_{l}\left(  \boldsymbol{v}_{l},\boldsymbol{h}_{l};\boldsymbol{\theta}%
_{l}\right)  =-%
%TCIMACRO{\dsum \limits_{j\in G_{l}}}%
%BeginExpansion
{\displaystyle\sum\limits_{j\in G_{l}}}
%EndExpansion
\text{ }%
%TCIMACRO{\dsum \limits_{k\in G_{l}}}%
%BeginExpansion
{\displaystyle\sum\limits_{k\in G_{l}}}
%EndExpansion
w_{k,j}^{l}v_{k}^{l}h_{j}^{l}-%
%TCIMACRO{\dsum \limits_{j\in G_{l}}}%
%BeginExpansion
{\displaystyle\sum\limits_{j\in G_{l}}}
%EndExpansion
a_{j}^{l}v_{j}^{l}-%
%TCIMACRO{\dsum \limits_{j\in G_{l}}}%
%BeginExpansion
{\displaystyle\sum\limits_{j\in G_{l}}}
%EndExpansion
b_{j}^{l}h_{j}^{l}\text{,}%
\]
where $\boldsymbol{v}_{l}\boldsymbol{=}\left[  v_{k}^{l}\right]  _{k\in G_{l}%
}$ is the state of the visible field, $\boldsymbol{h}_{l}\boldsymbol{=}\left[
h_{k}^{l}\right]  _{k\in G_{l}}$ is the state of the hidden field, and
$\boldsymbol{\theta}=\left(  w_{k,j},a_{j},b_{j}\right)  $. We assume that the
fields $\boldsymbol{v}_{l}$, $\boldsymbol{h}_{l}$\ are binary-valued, i.e.,
$\boldsymbol{v}_{l}$, $\boldsymbol{h}_{l}\in\left\{  0,1\right\}  ^{\#G_{l}}$,
where $\#G_{l}$ is the cardinality of $G_{l}$.

The corresponding Boltzmann distribution is given by%
\[
\boldsymbol{P}_{l}(\boldsymbol{v}_{l},\boldsymbol{h}_{l})=\frac{\exp\left(
-E_{l}\left(  \boldsymbol{v}_{l},\boldsymbol{h}_{l}\right)  \right)  }{%
%TCIMACRO{\dsum \limits_{\boldsymbol{v}_{l}^{\prime},\boldsymbol{h}_{l}%
%^{\prime}}}%
%BeginExpansion
{\displaystyle\sum\limits_{\boldsymbol{v}_{l}^{\prime},\boldsymbol{h}%
_{l}^{\prime}}}
%EndExpansion
\exp\left(  -E_{l}\left(  \boldsymbol{v}_{l}^{\prime},\boldsymbol{h}%
_{l}^{\prime}\right)  \right)  }.
\]
If the entries of the matrix $\left[  w_{k,j}^{l}\right]  $ do not depend on
the topology of $G_{l}$ neither on the group structure of $G_{l}$, then
$E_{l}\left(  \boldsymbol{v}_{l},\boldsymbol{h}_{l};\boldsymbol{\theta}%
_{l}\right)  $ defines a standard RBM. If the entries of the matrix $\left[
w_{k,j}\right]  $ depend on topology of $G_{l}$ and on the group structure of
$G_{l}$, then $\left[  w_{k,j}\right]  $ is a Parisi-type matrix, see e.g.
\cite{DKVK-Parisi} and the references therein. In this article we assume that
$w_{k,j}^{l}=w\left(  k-j;l\right)  $, in this case $E_{l}\left(
\boldsymbol{v}_{l},\boldsymbol{h}_{l};\boldsymbol{\theta}_{l}\right)  $ is the
energy functional of a new non-Archimedean convolutional deep belief network
(DBN), see e.g. \cite{Dong et al}, \cite{Honglak et al}. Since $G_{l}$ is an
additive group, the weight $w_{k,j}^{l}$ depends on one parameter $k-j\in
G_{l}$. More generally, $E_{l}$ has a translational symmetry: $k\rightarrow
k+i_{0}$, $j\rightarrow j+i_{0}$. We denote by $DBN(p,l,\boldsymbol{\theta
}_{l})$ the $p$-adic discrete deep belief network attached to $E_{l}\left(
\boldsymbol{v}_{l},\boldsymbol{h}_{l};\boldsymbol{\theta}_{l}\right)  $.

The $p$-adic convolutional DBNs are a particular type of DBNs. The
translational invariance implies that the $p$-adic DBNs have less parameters
than the standard DBN, see Section \ref{Section p-adic DBN}\ for further
details. Two fundamental questions come up immediately: can \ the $p$-adic
convolutional DBNs perform computations?; does the computational power of the
$p$-adic convolutional\ DBNs increase as the number of levels of the $G_{l}$
tree increases?. The answers to both questions is yes, see Theorems
\ref{Theorem_2}, \ref{Theorem_3}. These theorems constitute the main results
of this article.

We identify $G_{l}$ with a subset of $G_{l+1}$. Given an
$DBN(p,l,\boldsymbol{\theta}_{l})$ we construct another
$DBN(p,l+1,\boldsymbol{\theta}_{l},\boldsymbol{w}_{l+1},b_{j_{0}}^{l+1})$,
here $\boldsymbol{w}_{l+1}\in\mathbb{R}^{\#G_{l}}$, $b_{j_{0}}^{l+1}%
\in\mathbb{R}$, with an extra layer and an extra hidden unit, and with the
same visible units, whose energy functional $E_{l+1}(\boldsymbol{v}%
_{l+1},\boldsymbol{h}_{l+1};\boldsymbol{\theta}_{l},\boldsymbol{w}%
_{l+1},b_{j_{0}}^{l+1})=E_{l+1}(\boldsymbol{v}_{l},\boldsymbol{h}%
_{l+1};\boldsymbol{\theta}_{l},\boldsymbol{w}_{l+1},b_{j_{0}}^{l+1})$ is an
extension of $E_{l}\left(  \boldsymbol{v}_{l},\boldsymbol{h}_{l}%
;\boldsymbol{\theta}_{l}\right)  $, here $\boldsymbol{h}_{l+1}=\left[
\begin{array}
[c]{c}%
\boldsymbol{h}_{l}\\
h_{j_{0}}^{l+1}%
\end{array}
\right]  $, and $h_{j_{0}}^{l+1}$ is the extra hidden unit, see Section
\ref{SEction-Key-Construction}. In Theorem \ref{Theorem_2}, by adapting the
mathematical techniques introduced by Le Roux and Bengio in \cite{Le roux et
al 1}, we show that if $KL(\boldsymbol{Q}\left(  \boldsymbol{v}_{l}\right)
\mid\boldsymbol{P}_{l}\left(  \boldsymbol{v}_{l};\boldsymbol{\theta}%
_{l}\right)  )>0$, where $KL$ denotes the relative entropy between an
arbitrary probability distribution $\boldsymbol{Q}(\boldsymbol{v})$ on
$\left\{  0,1\right\}  ^{m}$ and $\boldsymbol{P}_{l}\left(  \boldsymbol{v}%
_{l};\boldsymbol{\theta}_{l}\right)  $, then there exists an
$DBN(p,l+1,\boldsymbol{\theta}_{l},\boldsymbol{w}_{l+1},b_{j_{0}}^{l+1})$
constructed from $DBN(p,l,\boldsymbol{\theta}_{l})$ by adding one layer with
marginal probability distribution $\boldsymbol{P}_{l+1}\left(  \boldsymbol{v}%
_{l};\boldsymbol{\theta}_{l},\boldsymbol{w}_{l+1},b_{j_{0}}^{l+1}\right)  $
satisfying%
\[
KL(\boldsymbol{Q}\left(  \boldsymbol{v}_{l}\right)  \mid\boldsymbol{P}%
_{l+1}\left(  \boldsymbol{v}_{l};\boldsymbol{\theta}_{l},\boldsymbol{w}%
_{l+1},b_{j_{0}}^{l+1}\right)  )<KL(\boldsymbol{Q}\left(  \boldsymbol{v}%
_{l}\right)  \mid\boldsymbol{P}_{l}\left(  \boldsymbol{v}_{l}%
;\boldsymbol{\theta}_{l}\right)  ),
\]
for some $\boldsymbol{\theta}_{l},\boldsymbol{w}_{l+1},b_{j_{0}}^{l+1}$. This
inequality implies that the $p$-adic discrete DBNs are universal
approximators. In Theorem \ref{Theorem_3}, we show that $\boldsymbol{Q}\left(
\boldsymbol{v}_{l}\right)  $ can be approximated arbitrarily well, in the
sense of the $KL$ distance, by an $DBN(p,l_{0}+k,\boldsymbol{\theta}_{l_{0}%
+k})$, where $k$ is the number of input vectors whose probability in not zero.

In Section \ref{Section2}, we identify a $p$-adic continuos DBN with a SFT
with energy functional $E\left(  \boldsymbol{v},\boldsymbol{h}\right)  $,
where $\boldsymbol{v},\boldsymbol{h}$ are general functions on the unit ball
$\mathbb{Z}_{p}$. We show that there is a dicretization method to obtain a
discrete energy functional $E_{l}\left(  \boldsymbol{v}_{l},\boldsymbol{h}%
_{l}\right)  $ on $G_{l}$, and thus a $p$-adic discrete DBN. The $p$-adic
discrete DBNs, include the classical ones, as well as the $p$-adic
convolutional\ DBNs considered here. Not all these networks can be considered
useful for engineering purposes, particularly those that do not have the
property of being universal approximators should be discarded. Then, our
Theorems \ref{Theorem_2}, \ref{Theorem_3} show that the study of the
correspondence between $p$-adic STFs and NNs\ is a promising field. Here, we
propose a research program that aims to provide a precise formulation of the
mentioned correspondence, see Section \ref{Section4}. In this work, we do not
consider applications or simulations of the $p$-adic discrete DBNs. This will
be discussed in a forthcoming publication \cite{Zuniga et al}. Neither,
perturbative calculations of correlation functions via Feynman diagrams. We
plan to discuss these matters in a future publication.

In the middle of the 80s the idea of using ultrametric spaces to describe the
states of complex systems, which naturally possess a hierarchical structure,
emerged in the works of Parisi, Frauenfelder, Stein, among others, see, e.g.,
\cite{DKVK-Parisi}, \cite{D-K-K-V}, \cite{Fraunfelder et al}, \cite{M-P-V},
\cite{R-T-V}. A central paradigm in physics of complex systems (for instance
proteins) asserts that the dynamics of such systems can be modeled as a random
walk in the energy landscape of the system, see, e.g., \cite{Fraunfelder et
al}, \cite{Kozyrev SV}, and the references therein.\ In protein physics, it is
regarded as one of the most profound ideas put forward to explain the nature
of distinctive life attributes. Typically these landscapes have a huge number
of local minima. It is clear that a description of the dynamics on such
landscapes requires an adequate approximation. By using, interbasin kinetics
methods, an energy landscape is approximated by an ultrametric space (a finite
rooted tree called a disconnectivity graph) and a function on this space
describing the distribution of the activation barriers, see, e.g.,
\cite{Becker et al}.

An ultrametric space $(\mathcal{M},d)$ is a metric space $\mathcal{M}$ with a
distance satisfying the strong triangle inequality $d(A,B)\leq\max\left\{
d\left(  A,C\right)  ,d\left(  B,C\right)  \right\}  $ for any three points
$A$, $B$, $C$ in $\mathcal{M}$. The field of $p$-adic numbers $\mathbb{Q}_{p}$
constitutes a central example of an ultrametric space. We use the term
`ultrametricity' to mean the emergence of ultrametric spaces in physical
models, and the term `hierarchical system' to mean a system whose states are
organized in a tree-like structure. A large class of these systems can be
modeled using ultrametric spaces. Ultrametric models have been applied in many
areas, including, brain and mental states models, relaxation of complex
systems, spin glasses, evolutionary dynamics, among other areas, see, e.g.,
\cite{Av-4}-\cite{Av-5}, \cite{D-K-K-V}, \cite{Khrenikov2A}%
-\cite{Khrennikov-Kozyrev}, \cite{Mukhamedov-1}-\cite{R-T-V}, \cite{V-V-Z},
\cite{Zunifa-RMP-2022}-\cite{Zuniga-LNM-2016}, and the references therein.

The Ising models over ultrametric spaces have been studied intensively, see,
e.g., \cite{DysonFreeman}, \cite{Gubser et al-1}, \cite{Khrennikov et al},
\cite{LM89}, \cite{Mis-4}, \cite{Mukhamedov-1}-\cite{Parisi-Sourlas},
\cite{Sinai}, \cite{Zuniga-JMP-2020}-\cite{Zuniga-JMP-2020-1} and the
references therein. An important motivation comes from the hierarchical Ising
model introduced in \cite{DysonFreeman}. The hierarchical Hamiltonian
introduced by Dyson in \cite{DysonFreeman} can be naturally studied in
$p$-adic spaces, see, e.g., \cite{LM89}, \cite{Gubser et al-1}. In
\cite{Parisi-Sourlas}, see also \cite{Khrennikov-Kozyrev}, Parisi and Sourlas
presented a $p$-adic formulation of replica symmetry breaking. Here, it is
important to say that the ultrametricity due to the replica breaking an the
ultrametricity of the $p$-adic field are different notions.

In \cite{Hinton et al}, see also \cite{Honglak et al},\ \cite{Dong et al},
Hinton et al. introduced the deep belief networks (DBNs), which are multilayer
hierarchical generative models constructed by stacking RBMs. The purpose of
this construction is to create a network whose neurons form a large tree-like
structure (a deep architecture). Since a binary RBM is a spin glass, we argue
that an DBN\ must correspond to a hierarchical spin glass. Models of such
systems can be constructed in an easy way by using $p$-adic numbers. In our
view, this is a new approach to understanding deep learning architectures.

The article is organized as follows.\ In Section \ref{Section1}, we review the
basic aspects of the $p$-adic analysis. In Section \ref{Section2}, we
introduce the $p$-adic RBMs and their discretizations. In Section
\ref{Section3}, we show that the discrete $p$-adic RBMs are universal
approximators. Finally, in the last section, we present a discussion of our
results compared with other related work, also we propose several open problems.

\section{\label{Section1} Basic facts on $p$-adic analysis}

In this section we fix the notation and collect some basic results on $p$-adic
analysis that we will use through the article. For a detailed exposition on
$p$-adic analysis the reader may consult \cite{A-K-S}, \cite{Taibleson},
\cite{V-V-Z}. For a quick review of $p$-adic analysis the reader may consult
\cite{Bocardo-Zuniga-2}, \cite{Leon-Zuniga}.

\subsection{The field of $p$-adic numbers}

Throughout this article $p$ will denote a prime number. The field of $p-$adic
numbers $\mathbb{Q}_{p}$ is defined as the completion of the field of rational
numbers $\mathbb{Q}$ with respect to the $p-$adic norm $|\cdot|_{p}$, which is
defined as
\[
|x|_{p}=%
\begin{cases}
0 & \text{if }x=0\\
p^{-\gamma} & \text{if }x=p^{\gamma}\dfrac{a}{b},
\end{cases}
\]
where $a$ and $b$ are integers coprime with $p$. The integer $\gamma
=ord_{p}(x)$ with $ord_{p}(0):=+\infty$, is called the\textit{\ }%
$p-$\textit{adic order of} $x$. The metric space $\left(  \mathbb{Q}%
_{p},\left\vert \cdot\right\vert _{p}\right)  $ is a complete ultrametric
space. Ultrametric means that $\left\vert x+y\right\vert _{p}\leq\max\left\{
\left\vert x\right\vert _{p},\left\vert y\right\vert _{p}\right\}  $. As a
topological space $\mathbb{Q}_{p}$\ is homeomorphic to a Cantor-like subset of
the real line, see e.g. \cite{A-K-S}, \cite{V-V-Z}.

Any $p-$adic number $x\neq0$ has a unique expansion of the form
\begin{equation}
x=p^{ord_{p}(x)}\sum_{j=0}^{\infty}x_{j}p^{j}, \label{expansion}%
\end{equation}
where $x_{j}\in\{0,1,2,\dots,p-1\}$ and $x_{0}\neq0$. It follows from
(\ref{expansion}), that any $x\in\mathbb{Q}_{p}\smallsetminus\left\{
0\right\}  $ can be represented uniquely as $x=p^{ord_{p}(x)}u\left(
x\right)  $ and $\left\vert x\right\vert _{p}=p^{-ord_{p}(x)}$.

\subsection{Topology of $\mathbb{Q}_{p}$}

For $r\in\mathbb{Z}$, denote by $B_{r}(a)=\{x\in\mathbb{Q}_{p};\left\vert
x-a\right\vert _{p}\leq p^{r}\}$ \textit{the ball of radius }$p^{r}$
\textit{with center at} $a\in\mathbb{Q}_{p}$, and take $B_{r}(0):=B_{r}$. The
ball $B_{0}$ equals \textit{the ring of }$p-$\textit{adic integers
}$\mathbb{Z}_{p} $. We also denote by $S_{r}(a)=\{x\in\mathbb{Q}%
_{p};\left\vert x-a\right\vert _{p}=p^{r}\}$ \textit{the sphere of radius
}$p^{r}$ \textit{with center at} $a\in\mathbb{Q}_{p}$, and take $S_{r}%
(0):=S_{r}$. We notice that $S_{0}^{1}=\mathbb{Z}_{p}^{\times}$ (the group of
units of $\mathbb{Z}_{p}$). The balls and spheres are both open and closed
subsets in $\mathbb{Q}_{p}$. In addition, two balls in $\mathbb{Q}_{p}$ are
either disjoint or one is contained in the other.

As a topological space $\left(  \mathbb{Q}_{p},\left\vert \cdot\right\vert
_{p}\right)  $ is totally disconnected, i.e. the only connected \ subsets of
$\mathbb{Q}_{p}$ are the empty set and the points. A subset of $\mathbb{Q}%
_{p}$ is compact if and only if it is closed and bounded in $\mathbb{Q}_{p}$,
see e.g. \cite[Section 1.3]{V-V-Z}, or \cite[Section 1.8]{A-K-S}. The balls
and spheres are compact subsets. Thus $\left(  \mathbb{Q}_{p},\left\vert
\cdot\right\vert _{p}\right)  $ is a locally compact topological space.

Since $(\mathbb{Q}_{p},+)$ is a locally compact topological group, there
exists a Haar measure $dx$, which is invariant under translations, i.e.
$d(x+a)=dx$. If we normalize this measure by the condition $\int
_{\mathbb{Z}_{p}}dx=1$, then $dx$ is unique. In a few occasions we use the
two-dimensional Haar measure $dxdy$ of the additive group $(\mathbb{Q}%
_{p}\times\mathbb{Q}_{p},+)$ normalized by the condition $\int_{\mathbb{Z}%
_{p}}\int_{\mathbb{Z}_{p}}dxdy=1$. For a quick review of the integration in
the $p$-adic framework the reader may consult \cite{Bocardo-Zuniga-2},
\cite{Leon-Zuniga} and the references therein.

\begin{notation}
We will use $\Omega\left(  p^{-r}\left\vert x-a\right\vert _{p}\right)  $ to
denote the characteristic function of the ball $B_{r}(a)$.
\end{notation}

\subsection{The Bruhat-Schwartz space}

A real-valued function $\varphi$ defined on $\mathbb{Q}_{p}$ is \textit{called
locally constant} if for any $x\in\mathbb{Q}_{p}$ there exist an integer
$l(x)\in\mathbb{Z}$ such that%
\begin{equation}
\varphi(x+x^{\prime})=\varphi(x)\text{ for any }x^{\prime}\in B_{l(x)}.
\label{local_constancy}%
\end{equation}
A function $\varphi:\mathbb{Q}_{p}\rightarrow\mathbb{C}$ is called a
\textit{Bruhat-Schwartz function (or a test function)} if it is locally
constant with compact support. Any test function can be represented as a
linear combination, with real coefficients, of characteristic functions of
balls. The $\mathbb{R}$-vector space of Bruhat-Schwartz functions is denoted
by $\mathcal{D}(\mathbb{Q}_{p})$. For $\varphi\in\mathcal{D}(\mathbb{Q}_{p})$,
the largest number $l=l(\varphi)$ satisfying (\ref{local_constancy}) is called
\textit{the exponent of local constancy (or the parameter of constancy) of}
$\varphi$. Let $U$ be an open subset of $%
%TCIMACRO{\U{211a} }%
%BeginExpansion
\mathbb{Q}
%EndExpansion
_{p}$, we denote by $\mathcal{D}(U)$ the $\mathbb{R}$-vector space of all test
functions with support in $U$. For instance $\mathcal{D}(\mathbb{Z}_{p})$ is
the $\mathbb{R}$-vector space of all test functions with supported in $\ $the
unit ball $\mathbb{Z}_{p}$. A function $\varphi$ in $\mathcal{D}%
(\mathbb{Z}_{p})$ can be written as%
\[
\varphi\left(  x\right)  =%
%TCIMACRO{\dsum \limits_{j=1}^{M}}%
%BeginExpansion
{\displaystyle\sum\limits_{j=1}^{M}}
%EndExpansion
\varphi\left(  \widetilde{x}_{j}\right)  \Omega\left(  p^{r_{j}}\left\vert
x-\widetilde{x}_{j}\right\vert _{p}\right)  ,
\]
where the $\widetilde{x}_{j}$, $j=1,\ldots,M$, are points in $\mathbb{Z}_{p}$,
the $r_{j}$, $j=1,\ldots,M$, are integers, and $\Omega\left(  p^{r_{j}%
}\left\vert x-\widetilde{x}_{j}\right\vert _{p}\right)  $ denotes the
characteristic function of the ball $B_{-r_{j}}(\widetilde{x}_{j}%
)=\widetilde{x}_{j}+p^{r_{j}}\mathbb{Z}_{p}$.

\section{\label{Section2}A class of non-Archimedean statistical field theories
and their discretizations}

In this section we introduce a family of $p$-adic SFTs that we require in this
article. A central difference between the classical SFTs and the
non-Archimedean counterparts is that the discretization process of the
non-Archimedean ones is very simple, and the convergence of the discrete
theories to continuous theories can be formulated in rigorous mathematical
way, in a large number of cases, see e.g. \cite{Zunifa-RMP-2022} and the
references therein. For a mathematical exposition of the non-Archimedean
$\phi^{4}$-QFTs, the reader may consult \cite{Zunifa-RMP-2022}, see also
\cite{Arroyo-Zuniga}, and the references therein.

\subsection{A class of non-Archimedean statistical field theories}

We fix $a\left(  x\right)  $, $b(x)\in\mathcal{D}(\mathbb{Z}_{p})$, and a
function $w(x,y):\mathbb{Z}_{p}\times\mathbb{Z}_{p}\rightarrow\mathbb{R}$. In
this section, $w(x,y)$ is a test function of two variables. But the results
can be easily extended to translational invariant kernels of type
$w(x,y)=w\left(  x-y\right)  $. A $p$\textit{-adic continuous deep belief
network (or a }$p$\textit{-adic continuous DBN)} is a statistical field
theory$\ \left\{  \boldsymbol{v},\boldsymbol{h}\right\}  $ in $\mathcal{D}%
(\mathbb{Z}_{p})$. The function $\boldsymbol{v}(x)\in\mathcal{D}%
(\mathbb{Z}_{p})$ is called the \textit{visible field} and the function
$\boldsymbol{h}(x)\in\mathcal{D}(\mathbb{Z}_{p})$ is called the \textit{hidden
field}. The field $\left\{  \boldsymbol{v},\boldsymbol{h}\right\}  $ performs
thermal fluctuations, assuming that the expectation value of the field is
zero, the fluctuations take place around zero. The size of the fluctuations is
controlled by an energy functional (or action) of the form%
\begin{equation}
E\left(  \boldsymbol{v},\boldsymbol{h}\right)  =-%
%TCIMACRO{\diint \limits_{\mathbb{Z}_{p}\times\mathbb{Z}_{p}}}%
%BeginExpansion
{\displaystyle\iint\limits_{\mathbb{Z}_{p}\times\mathbb{Z}_{p}}}
%EndExpansion
\boldsymbol{h}\left(  y\right)  w\left(  x,y\right)  \boldsymbol{v}\left(
x\right)  dxdy-%
%TCIMACRO{\dint \limits_{\mathbb{Z}_{p}}}%
%BeginExpansion
{\displaystyle\int\limits_{\mathbb{Z}_{p}}}
%EndExpansion
a(x)\boldsymbol{v}\left(  x\right)  dx-%
%TCIMACRO{\dint \limits_{\mathbb{Z}_{p}}}%
%BeginExpansion
{\displaystyle\int\limits_{\mathbb{Z}_{p}}}
%EndExpansion
b(x)\boldsymbol{h}\left(  x\right)  dx. \label{Energy_Functional}%
\end{equation}
Along the article, we assume that the fields $\boldsymbol{v}$, $\boldsymbol{h}%
$ are binary valued functions, i.e. $\boldsymbol{v}$, $\boldsymbol{h}:$
$\mathbb{Z}_{p}\rightarrow\left\{  0,1\right\}  \subset\mathbb{R}$, then
$E\left(  \boldsymbol{v},\boldsymbol{h}\right)  $ is the energy functional of
a $p$-adic continuous spin glass.

All thermodynamic properties of the system are described\ by the partition
function of the fluctuating field, which is defined as
\[
Z^{\text{phys}}=%
%TCIMACRO{\dint }%
%BeginExpansion
{\displaystyle\int}
%EndExpansion
d\boldsymbol{v}d\boldsymbol{h}\text{ }e^{-\frac{E(\boldsymbol{v}%
,\boldsymbol{h})}{K_{B}T}},
\]
where $K_{B}$ is the Boltzmann constant and $T$ is the temperature. We
normalize in such a way that $K_{B}T=1$. The measure $d\boldsymbol{v}%
d\boldsymbol{h}$ is ill-defined. It is expected that such measure can be
defined rigorously\ by a limit process.

The statistical field theory corresponding to the energy functional
(\ref{Energy_Functional}) is the ill-defined probability measure%
\[
\boldsymbol{P}^{\text{phys}}(\boldsymbol{v},\boldsymbol{h})=d\boldsymbol{v}%
d\boldsymbol{h}\frac{\exp\left(  -E(\boldsymbol{v},\boldsymbol{h})\right)
}{Z^{\text{phys}}}%
\]
on the space of functions $\mathcal{D}(\mathbb{Z}_{p})\times\mathcal{D}%
(\mathbb{Z}_{p})$.

The information about the local properties of the system is contained in the
\textit{correlation functions }$G_{\mathbb{I},\mathbb{K}}^{\left(  n\right)
}\left(  x_{1},\ldots,x_{n}\right)  $ of the field $\left\{  \boldsymbol{v}%
,\boldsymbol{h}\right\}  $: for $n\geq1$, and two disjoint subsets
$\mathbb{I}$, $\mathbb{K}\subset\left\{  1,2,\ldots,n\right\}  $, with
$\mathbb{I}%
%TCIMACRO{\tcoprod }%
%BeginExpansion
{\textstyle\coprod}
%EndExpansion
\mathbb{K}=\left\{  1,2,\ldots,n\right\}  $, we set
\begin{align*}
G_{\mathbb{I},\mathbb{K}}^{\left(  n\right)  }\left(  x_{1},\ldots
,x_{n}\right)   &  =\left\langle
%TCIMACRO{\dprod \limits_{i\in\mathbb{I}}}%
%BeginExpansion
{\displaystyle\prod\limits_{i\in\mathbb{I}}}
%EndExpansion
\boldsymbol{v}\left(  x_{i}\right)  \text{ }%
%TCIMACRO{\dprod \limits_{j\in\mathbb{K}}}%
%BeginExpansion
{\displaystyle\prod\limits_{j\in\mathbb{K}}}
%EndExpansion
\boldsymbol{h}\left(  x_{j}\right)  \right\rangle \\
&  =\frac{1}{Z^{\text{phys}}}%
%TCIMACRO{\dint }%
%BeginExpansion
{\displaystyle\int}
%EndExpansion
d\boldsymbol{v}d\boldsymbol{h}\text{ }%
%TCIMACRO{\dprod \limits_{i\in\mathbb{I}}}%
%BeginExpansion
{\displaystyle\prod\limits_{i\in\mathbb{I}}}
%EndExpansion
\boldsymbol{v}\left(  x_{i}\right)  \text{ }%
%TCIMACRO{\dprod \limits_{j\in\mathbb{K}}}%
%BeginExpansion
{\displaystyle\prod\limits_{j\in\mathbb{K}}}
%EndExpansion
\boldsymbol{h}\left(  x_{j}\right)  \text{ }e^{-E(\boldsymbol{v}%
,\boldsymbol{h})}.
\end{align*}
These functions are also called the $n$-point Green functions. To study of
these functions, one introduces two auxiliary external fields $J_{0}(x),$
$J_{1}(x)\in\mathcal{D}(\mathbb{Z}_{p})$ called \textit{currents, }and adds to
the energy functional $E$ as a linear interaction energy of these currents
with the field $\left\{  \boldsymbol{v},\boldsymbol{h}\right\}  $,%
\[
E_{\text{source}}(\boldsymbol{v},\boldsymbol{h},J_{0},J_{1})=-%
%TCIMACRO{\dint \limits_{\mathbb{Z}_{p}}}%
%BeginExpansion
{\displaystyle\int\limits_{\mathbb{Z}_{p}}}
%EndExpansion
J_{0}(x)\boldsymbol{v}\left(  x\right)  dx-%
%TCIMACRO{\dint \limits_{\mathbb{Z}_{p}}}%
%BeginExpansion
{\displaystyle\int\limits_{\mathbb{Z}_{p}}}
%EndExpansion
J_{1}(x)\boldsymbol{h}\left(  x\right)  dx,
\]
and the energy functional is $E(\boldsymbol{v},\boldsymbol{h},J_{0}%
,J_{1})=E\left(  \boldsymbol{v},\boldsymbol{h}\right)  +E_{\text{source}%
}(\boldsymbol{v},\boldsymbol{h},J_{0},J_{1})$. The partition function formed
with this energy is%
\[
Z(J_{0},J_{1})=\frac{1}{Z^{\text{phys}}}%
%TCIMACRO{\dint }%
%BeginExpansion
{\displaystyle\int}
%EndExpansion
d\boldsymbol{v}d\boldsymbol{h}\text{ }e^{-E(\boldsymbol{v},\boldsymbol{h}%
,J_{0},J_{1})}.
\]
The functional derivatives of $Z(J_{0},J_{1})$ with respect to $J_{0}(x)$,
$J_{1}(x)$ evaluated at $J_{0}=0$, $J_{1}=0$\ give the correlation functions
of the system:%
\[
G_{\mathbb{I},\mathbb{K}}^{\left(  n\right)  }\left(  x_{1},\ldots
,x_{n}\right)  =\left[
%TCIMACRO{\tprod \limits_{i\in\mathbb{I}}}%
%BeginExpansion
{\textstyle\prod\limits_{i\in\mathbb{I}}}
%EndExpansion
\frac{\delta}{\delta J_{0}\left(  x_{i}\right)  }\text{ }%
%TCIMACRO{\tprod \limits_{j\in\mathbb{K}}}%
%BeginExpansion
{\textstyle\prod\limits_{j\in\mathbb{K}}}
%EndExpansion
\frac{\delta}{\delta J_{1}\left(  x_{j}\right)  }Z(J_{0},J_{1})\right]
_{\substack{J_{0}=0\\J_{1}=0}}.
\]
The functional $Z(J_{0},J_{1})$ is called the \textit{generating functional}
of the theory.

\subsection{Discretization of the energy functional}

For $l\geq1$, we \ set $G_{l}:=\mathbb{Z}_{p}/p^{l}\mathbb{Z}_{p}$. We use a
fixed system of representatives of the form
\[
i\boldsymbol{=}i_{0}+i_{1}p+\ldots+i_{l-1}p^{l-1},
\]
where the $i_{k}$s\ are $p$-adic digits, for the elements of $G_{l}$. We
denote by $\mathcal{D}^{l}(\mathbb{Z}_{p})$ the $\mathbb{R}$-vector space of
all test functions of the form%
\begin{equation}
\varphi\left(  x\right)  =%
%TCIMACRO{\tsum \limits_{i\in G_{l}}}%
%BeginExpansion
{\textstyle\sum\limits_{i\in G_{l}}}
%EndExpansion
\varphi\left(  i\right)  \Omega\left(  p^{l}\left\vert x-i\right\vert
_{p}\right)  \text{, \ }\varphi\left(  i\right)  \in\mathbb{R}\text{,}
\label{Eq_repre}%
\end{equation}
where $\Omega\left(  p^{l}\left\vert x-i\right\vert _{p}\right)  $ denotes the
characteristic function of the ball $i+p^{l}\mathbb{Z}_{p}$. Notice that
$\varphi$ is supported on $\mathbb{Z}_{p}$ and that $\mathcal{D}%
^{l}(\mathbb{Z}_{p})$ is a finite dimensional vector space spanned by the
basis
\begin{equation}
\left\{  \Omega\left(  p^{l}\left\vert x-i\right\vert _{p}\right)  \right\}
_{i\in G_{l}}. \label{Basis}%
\end{equation}
By identifying $\varphi\in\mathcal{D}^{l}(\mathbb{Z}_{p})$ with the column
vector $\left[  \varphi\left(  i\right)  \right]  _{\boldsymbol{i}\in G_{l}%
}\in\mathbb{R}^{\#G_{l}}$, we get that $\mathcal{D}^{l}(\mathbb{Z}_{p})$ is
isomorphic to $\mathbb{R}^{\#G_{l}}$ endowed with the norm $\left\Vert \left[
\varphi\left(  i\right)  \right]  _{\boldsymbol{i}\in G_{l}^{N}}\right\Vert
=\max_{i\in G_{l}}\left\vert \varphi\left(  i\right)  \right\vert $.
Furthermore,
\[
\mathcal{D}^{l}(\mathbb{Z}_{p})\hookrightarrow\mathcal{D}^{l+1}(\mathbb{Z}%
_{p})\hookrightarrow\mathcal{D}(\mathbb{Z}_{p}),
\]
where $\hookrightarrow$ denotes a continuous embedding.

A discretization $E_{l}$ of the energy functional $E$ is obtained by
restricting $\boldsymbol{v},\boldsymbol{h}$ to $\mathcal{D}^{l}(\mathbb{Z}%
_{p})$, i.e. by taking
\[
\boldsymbol{v}\left(  x\right)  =%
%TCIMACRO{\tsum \limits_{i\in G_{l}}}%
%BeginExpansion
{\textstyle\sum\limits_{i\in G_{l}}}
%EndExpansion
\boldsymbol{v}\left(  i\right)  \Omega\left(  p^{l}\left\vert x-i\right\vert
_{p}\right)  \text{, \ }\boldsymbol{h}\left(  x\right)  =%
%TCIMACRO{\tsum \limits_{i\in G_{l}}}%
%BeginExpansion
{\textstyle\sum\limits_{i\in G_{l}}}
%EndExpansion
\boldsymbol{h}\left(  i\right)  \Omega\left(  p^{l}\left\vert x-i\right\vert
_{p}\right)  .
\]
When $\boldsymbol{v},\boldsymbol{h}\in\mathcal{D}^{l}(\mathbb{Z}_{p})$, we use
use the following identifications:%
\[
\boldsymbol{v}_{l}\boldsymbol{=}\left[  \boldsymbol{v}\left(  i\right)
\right]  _{i\in G_{l}}\text{, }\boldsymbol{h}_{l}\boldsymbol{=}\left[
\boldsymbol{h}\left(  i\right)  \right]  _{i\in G_{l}}.
\]
There are two different types of discrete functionals $E_{l}$ according if
$w(x,y)$ is a test function \ of two variables, or if $w(x,y)=w\left(
x-y\right)  $ is a translational invariant test function of one variable.

\subsection{Standard restricted Boltzmann machines}

Assume that $w(x,y)$ is a test function. Since $w(x,y)$ is locally constant,
\[
w(x,y)\Omega\left(  p^{l}\left\vert x-i\right\vert _{p}\right)  \Omega\left(
p^{l}\left\vert x-j\right\vert _{p}\right)  =w(i,j)\Omega\left(
p^{l}\left\vert x-i\right\vert _{p}\right)  \Omega\left(  p^{l}\left\vert
x-j\right\vert _{p}\right)  ,
\]
for $l$ sufficiently large, and%
\begin{align*}
&
%TCIMACRO{\diint \limits_{\mathbb{Z}_{p}\times\mathbb{Z}_{p}}}%
%BeginExpansion
{\displaystyle\iint\limits_{\mathbb{Z}_{p}\times\mathbb{Z}_{p}}}
%EndExpansion
w\left(  x,y\right)  \Omega\left(  p^{l}\left\vert x-i\right\vert _{p}\right)
\Omega\left(  p^{l}\left\vert y-j\right\vert _{p}\right)  dxdy\\
&  =w(i,j)\left(  \text{ }%
%TCIMACRO{\dint \limits_{i+p^{l}\mathbb{Z}_{p}}}%
%BeginExpansion
{\displaystyle\int\limits_{i+p^{l}\mathbb{Z}_{p}}}
%EndExpansion
dx\right)  \left(  \text{ }%
%TCIMACRO{\dint \limits_{j+p^{l}\mathbb{Z}_{p}}}%
%BeginExpansion
{\displaystyle\int\limits_{j+p^{l}\mathbb{Z}_{p}}}
%EndExpansion
dy\right)  =p^{-2l}w(i,j).
\end{align*}
By a similar argument, for $l$ sufficiently large, we have%
\begin{align*}%
%TCIMACRO{\dint \limits_{\mathbb{Z}_{p}}}%
%BeginExpansion
{\displaystyle\int\limits_{\mathbb{Z}_{p}}}
%EndExpansion
a(x)\boldsymbol{v}(x)dx  &  =p^{-l}%
%TCIMACRO{\dsum \limits_{\boldsymbol{i}\in G_{l}}}%
%BeginExpansion
{\displaystyle\sum\limits_{\boldsymbol{i}\in G_{l}}}
%EndExpansion
a(i)\boldsymbol{v}\left(  i\right)  \text{,}\\%
%TCIMACRO{\dint \limits_{\mathbb{Z}_{p}}}%
%BeginExpansion
{\displaystyle\int\limits_{\mathbb{Z}_{p}}}
%EndExpansion
b(x)\boldsymbol{h}(x)dx  &  =p^{-l}%
%TCIMACRO{\dsum \limits_{\boldsymbol{i}\in G_{l}}}%
%BeginExpansion
{\displaystyle\sum\limits_{\boldsymbol{i}\in G_{l}}}
%EndExpansion
b(i)\boldsymbol{h}\left(  i\right)  .
\end{align*}
Therefore, for $l$ sufficiently large,%
\[
E_{l}\left(  \boldsymbol{v}_{l},\boldsymbol{h}_{l}\right)  =-p^{-2l}%
%TCIMACRO{\dsum \limits_{i,\text{ }j\in G_{l}}}%
%BeginExpansion
{\displaystyle\sum\limits_{i,\text{ }j\in G_{l}}}
%EndExpansion
\boldsymbol{v}\left(  i\right)  w\left(  i,j\right)  \boldsymbol{h}\left(
j\right)  -p^{-l}%
%TCIMACRO{\dsum \limits_{i\in G_{l}}}%
%BeginExpansion
{\displaystyle\sum\limits_{i\in G_{l}}}
%EndExpansion
a(i)\boldsymbol{v}\left(  i\right)  -p^{-l}%
%TCIMACRO{\dsum \limits_{\boldsymbol{i}\in G_{l}}}%
%BeginExpansion
{\displaystyle\sum\limits_{\boldsymbol{i}\in G_{l}}}
%EndExpansion
b(i)\boldsymbol{h}\left(  i\right)  .
\]
By taking
\[
v_{i}^{l}:=\boldsymbol{v}\left(  i\right)  \text{, }h_{i}^{l}:=\boldsymbol{h}%
\left(  i\right)  \text{, }w_{i,j}^{l}:=p^{-2l}w\left(  i,j\right)  \text{,
\ }a_{i}^{l}:=p^{-l}a(i)\text{, }b_{i}^{l}:=p^{-l}b(i),
\]%
\[
\boldsymbol{v}_{l}\boldsymbol{=}\left[  v_{i}^{l}\right]  _{i\in G_{l}}\text{,
}\boldsymbol{h}_{l}\boldsymbol{=}\left[  h_{i}^{l}\right]  _{i\in G_{l}%
}\text{, }\boldsymbol{w}_{l}=\left[  w_{i,j}^{l}\right]  _{i,j\in G_{l}%
}\text{, }\boldsymbol{a}_{l}=\left[  a_{i}^{l}\right]  _{\boldsymbol{i}\in
G_{l}}\text{, }\boldsymbol{b}_{l}=\left[  b_{i}^{l}\right]  _{i\in G_{l}}%
\]
we have%
\begin{equation}
E_{l}\left(  \boldsymbol{v}_{l},\boldsymbol{h}_{l}\right)  =-%
%TCIMACRO{\dsum \limits_{i,\text{ }j\in G_{l}}}%
%BeginExpansion
{\displaystyle\sum\limits_{i,\text{ }j\in G_{l}}}
%EndExpansion
v_{i}^{l}w_{i,j}^{l}h_{j}^{l}-%
%TCIMACRO{\dsum \limits_{i\in G_{l}}}%
%BeginExpansion
{\displaystyle\sum\limits_{i\in G_{l}}}
%EndExpansion
a_{i}^{l}v_{i}^{l}-%
%TCIMACRO{\dsum \limits_{i\in G_{l}}}%
%BeginExpansion
{\displaystyle\sum\limits_{i\in G_{l}}}
%EndExpansion
b_{i}^{l}h_{i}^{l}, \label{Energy_l_2}%
\end{equation}
which is the energy functional of a standard restricted Boltzman machine. Here
it is very relevant to notice that the energy functional $E_{l}\left(
\boldsymbol{v},\boldsymbol{h}\right)  $ does not depend on the topology of the
metric space $\left(  G_{l},\left\vert \cdot\right\vert _{p}\right)  $ neither
on the group structure $\left(  G_{l},+\right)  $. The Boltzmann distribution
is given by%
\[
\boldsymbol{P}_{l}\left(  \boldsymbol{v}_{l},\boldsymbol{h}_{l}\right)
=\frac{\exp\left(  -E_{l}\left(  \boldsymbol{v}_{l},\boldsymbol{h}_{l}\right)
\right)  }{Z_{l}},
\]
where $Z_{l}=\sum_{\boldsymbol{v}_{l},\boldsymbol{h}_{l}}\exp\left(
-E_{l}\left(  \boldsymbol{v}_{l},\boldsymbol{h}_{l}\right)  \right)  $.

Now, any standard RBM with $n$ visible nodes $v_{i}$, $i=1,\ldots,n$ and $m$
hidden nodes $h_{\boldsymbol{i}}$, $i=1,\ldots,m$ can be realized as $p$-adic
discrete RBM of type (\ref{Energy_l_2}), by choosing $p$ and $l$ satisfying
$n\leq p^{l}$, $m\leq p^{l}$ and taking $v_{i}^{l}=v_{i}$ for $1\leq i\leq n$,
$v_{i}^{l}=0$ for $n+1\leq i\leq p^{l}$, and $h_{i}^{l}=h_{i}$ or $1\leq i\leq
m$, $h_{i}^{l}=0$ for $m+1\leq i\leq p^{l}$.

The number of the $w_{i,j}^{l}$ parameters is $\left(  \#G_{l}\right)  ^{2}$,
the number of the $a_{i}$ parameters is $\#G_{l}$, and the number of the
$h_{i}$ parameters is $\#G_{l}$, and \ consequently the total number of
parameters is
\[
\left(  \#G_{l}\right)  ^{2}+2\left(  \#G_{l}\right)  ,
\]
which is quadratic \ in the cardinality of $G_{l}$.

\subsection{\label{Section p-adic DBN}$p$-adic discrete deep belief networks}

We now consider the case in which $w\left(  x,y\right)  =w\left(  x-y\right)
$ is test function of one variable. In this case the energy functional
$E_{l}\left(  \boldsymbol{v}_{l},\boldsymbol{h}_{l}\right)  $ depends on the
topology of the metric space $\left(  G_{l},\left\vert \cdot\right\vert
_{p}\right)  $ and on the group structure $\left(  G_{l},+\right)  $.

We first notice that%
\begin{align}%
%TCIMACRO{\diint \limits_{\mathbb{Z}_{p}\times\mathbb{Z}_{p}}}%
%BeginExpansion
{\displaystyle\iint\limits_{\mathbb{Z}_{p}\times\mathbb{Z}_{p}}}
%EndExpansion
&  w\left(  x-y_{p}\right)  \Omega\left(  p^{l}\left\vert x-i\right\vert
_{p}\right)  \Omega\left(  p^{l}\left\vert y-j\right\vert _{p}\right)
dxdy\label{Formula_I_w}\\
&  =\left\{
\begin{array}
[c]{ll}%
p^{-2l}w\left(  i-j\right)  & \text{if \ }i\neq j\\
& \\
p^{-2l}w(0) & \text{if }i=j,
\end{array}
\right. \nonumber
\end{align}
for $l$ sufficiently large.

Now, by using that $a\left(  x\right)  $, $b(x)$ are test functions supported
in the unit ball, and taking $l$ sufficiently large, we have
\begin{align*}
a\left(  x\right)  \Omega\left(  p^{l}\left\vert x-i\right\vert _{p}\right)
&  =a\left(  i\right)  \Omega\left(  p^{l}\left\vert x-i\right\vert
_{p}\right)  \text{, }\\
b(x)\Omega\left(  p^{l}\left\vert x-i\right\vert _{p}\right)   &
=b(i)\Omega\left(  p^{l}\left\vert x-i\right\vert _{p}\right)  \text{,}%
\end{align*}
and consequently
\begin{gather*}
E_{l}\left(  \boldsymbol{v}_{l},\boldsymbol{h}_{l}\right)  =-p^{-2l}%
%TCIMACRO{\dsum \limits_{\substack{i,\text{ }j\in G_{l}\\i\neq j}}}%
%BeginExpansion
{\displaystyle\sum\limits_{\substack{i,\text{ }j\in G_{l}\\i\neq j}}}
%EndExpansion
\boldsymbol{v}\left(  i\right)  w\left(  i-j\right)  \boldsymbol{h}\left(
j\right) \\
-p^{-2l}w(0)%
%TCIMACRO{\dsum \limits_{i\in G_{l}}}%
%BeginExpansion
{\displaystyle\sum\limits_{i\in G_{l}}}
%EndExpansion
\boldsymbol{v}\left(  i\right)  \boldsymbol{h}\left(  i\right)  -p^{-l}%
%TCIMACRO{\dsum \limits_{i\in G_{l}}}%
%BeginExpansion
{\displaystyle\sum\limits_{i\in G_{l}}}
%EndExpansion
a\left(  i\right)  \boldsymbol{v}\left(  i\right)  -p^{-l}%
%TCIMACRO{\dsum \limits_{i\in G_{l}}}%
%BeginExpansion
{\displaystyle\sum\limits_{i\in G_{l}}}
%EndExpansion
b\left(  i\right)  \boldsymbol{h}\left(  i\right)  .
\end{gather*}
By taking $v_{i}^{l}=\boldsymbol{v}\left(  i\right)  $, $h_{i}^{l}%
=\boldsymbol{h}\left(  i\right)  $,
\[
w_{i-j}^{l}=\left\{
\begin{array}
[c]{ll}%
p^{-2l}w\left(  i-j\right)  & \text{if \ }i\neq j\\
& \\
p^{-2l}w(0) & \text{if }i=j.
\end{array}
\right.
\]
$a_{i}^{l}=p^{-l}a\left(  i\right)  $, $b_{i}^{l}=p^{-l}b\left(  i\right)  $,
for $i$, $j\in G_{l}$, and $\boldsymbol{\theta}_{l}=\left\{  w_{i,j}^{l}%
,a_{i}^{l},b_{i}^{l}\right\}  $. Then, for $l$ sufficiently large,%
\begin{equation}
E_{l}\left(  \boldsymbol{v}_{l},\boldsymbol{h}_{l};\boldsymbol{\theta}%
_{l}\right)  =-%
%TCIMACRO{\dsum \limits_{i,\text{ }j\in G_{l}}}%
%BeginExpansion
{\displaystyle\sum\limits_{i,\text{ }j\in G_{l}}}
%EndExpansion
w_{i-j}^{l}v_{i}^{l}h_{j}^{l}-%
%TCIMACRO{\dsum \limits_{i\in G_{l}}}%
%BeginExpansion
{\displaystyle\sum\limits_{i\in G_{l}}}
%EndExpansion
a_{i}^{l}v_{i}^{l}-%
%TCIMACRO{\dsum \limits_{i\in G_{l}}}%
%BeginExpansion
{\displaystyle\sum\limits_{i\in G_{l}}}
%EndExpansion
b_{i}^{l}h_{i}^{l} \label{Energy_l}%
\end{equation}
Since $\left(  G_{l},+\right)  $ is an additive group, for $l$ sufficiently
large, we have%
\begin{equation}
E_{l}\left(  \boldsymbol{v}_{l},\boldsymbol{h}_{l};\boldsymbol{\theta}%
_{l}\right)  =-%
%TCIMACRO{\dsum \limits_{j\in G_{l}}}%
%BeginExpansion
{\displaystyle\sum\limits_{j\in G_{l}}}
%EndExpansion
\text{ }%
%TCIMACRO{\dsum \limits_{k\in G_{l}}}%
%BeginExpansion
{\displaystyle\sum\limits_{k\in G_{l}}}
%EndExpansion
w_{k}^{l}v_{j+k}^{l}h_{j}^{l}-%
%TCIMACRO{\dsum \limits_{j\in G_{l}}}%
%BeginExpansion
{\displaystyle\sum\limits_{j\in G_{l}}}
%EndExpansion
a_{j}^{l}\boldsymbol{v}_{j}^{l}-%
%TCIMACRO{\dsum \limits_{j\in G_{l}}}%
%BeginExpansion
{\displaystyle\sum\limits_{j\in G_{l}}}
%EndExpansion
b_{j}^{l}h_{j}^{l}\text{.} \label{Energy_lI}%
\end{equation}
The total number of parameters of this type of networks is%
\[
3\left(  \#G_{l}\right)  ,
\]
which is linear in the cardinality of $G_{l}$. In this type of network, the
visible and hidden states are functions on $G_{l}$. Only the vertices at the
top level of the tree $G_{l}$ are allowed to have states. The rest of the
vertices in the tree codify the hierarchical relations between the states. On
the other hand, in a standard deep belief network on $G_{l}$ all the vertices
can have states.

\begin{remark}
If $w\left(  x,y\right)  =w\left(  \left\vert x-y\right\vert _{p}\right)  $ is
a radial function, the energy functional takes the following form:
\begin{gather*}
E_{l}\left(  \boldsymbol{v}_{l},\boldsymbol{h}_{l}\right)  =-p^{-2l}%
%TCIMACRO{\dsum \limits_{\substack{i,\text{ }j\in G_{l}\\i\neq j}}}%
%BeginExpansion
{\displaystyle\sum\limits_{\substack{i,\text{ }j\in G_{l}\\i\neq j}}}
%EndExpansion
\boldsymbol{v}\left(  i\right)  w\left(  \left\vert i-j\right\vert
_{p}\right)  \boldsymbol{h}\left(  j\right) \\
-p^{-l}\left(  \text{ }%
%TCIMACRO{\dint \limits_{p^{l}\mathbb{Z}_{p}}}%
%BeginExpansion
{\displaystyle\int\limits_{p^{l}\mathbb{Z}_{p}}}
%EndExpansion
w\left(  \left\vert z\right\vert _{p}\right)  dz\right)
%TCIMACRO{\dsum \limits_{i\in G_{l}}}%
%BeginExpansion
{\displaystyle\sum\limits_{i\in G_{l}}}
%EndExpansion
\boldsymbol{v}\left(  i\right)  \boldsymbol{h}\left(  i\right)  -p^{-l}%
%TCIMACRO{\dsum \limits_{i\in G_{l}}}%
%BeginExpansion
{\displaystyle\sum\limits_{i\in G_{l}}}
%EndExpansion
a\left(  i\right)  \boldsymbol{v}\left(  i\right)  -p^{-l}%
%TCIMACRO{\dsum \limits_{i\in G_{l}}}%
%BeginExpansion
{\displaystyle\sum\limits_{i\in G_{l}}}
%EndExpansion
b\left(  i\right)  \boldsymbol{h}\left(  i\right)  .
\end{gather*}
In this case the network depends on $l+1+2\left(  \#G_{l}\right)  $
parameters. From now on, we focus on networks having an energy functional \ of
the form (\ref{Energy_lI}).
\end{remark}

\subsubsection{Boltzmann probability distributions}

From now on, we set $\boldsymbol{v}_{l}=\left[  v_{i}^{l}\right]  _{i\in
G_{l}}$, $\boldsymbol{h}_{l}=\left[  h_{i}^{l}\right]  _{i\in G_{l}}$,
$\boldsymbol{w}_{l}=\left[  w_{i}^{l}\right]  _{i\in G_{l}}$, $\boldsymbol{a}%
_{l}=\left[  a_{i}^{l}\right]  _{i\in G_{l}}$, $\boldsymbol{b}_{l}=\left[
b_{i}^{l}\right]  _{i\in G_{l}}$, $\boldsymbol{\theta}_{l}=(\boldsymbol{w}%
_{l},\boldsymbol{a}_{l},\boldsymbol{b}_{l})$. We warn the reader that, for the
sake of simplicity, the dependence of the $\boldsymbol{\theta}_{l}$ parameters
is omitted in most of the formulas. We associate to $E_{l}\left(
\boldsymbol{v}_{l},\boldsymbol{h}_{l};\boldsymbol{\theta}_{l}\right)  $ the
Boltzmann probability distribution%
\begin{equation}
\boldsymbol{P}_{l}(\boldsymbol{v}_{l},\boldsymbol{h}_{l})=\frac{\exp\left(
-E_{l}\left(  \boldsymbol{v}_{l},\boldsymbol{h}_{l}\right)  \right)  }{Z_{l}},
\label{Joint_Probability}%
\end{equation}
where%
\[
Z_{l}=%
%TCIMACRO{\dsum \limits_{\boldsymbol{v}_{l},\boldsymbol{h}_{l}}}%
%BeginExpansion
{\displaystyle\sum\limits_{\boldsymbol{v}_{l},\boldsymbol{h}_{l}}}
%EndExpansion
\exp\left(  -E_{l}\left(  \boldsymbol{v}_{l},\boldsymbol{h}_{l}\right)
\right)  =%
%TCIMACRO{\dsum \limits_{i,j\in G_{l}}}%
%BeginExpansion
{\displaystyle\sum\limits_{i,j\in G_{l}}}
%EndExpansion
\exp\left(  -E_{l}\left(  \left[  v_{i}^{l}\right]  _{i\in G_{l}}\text{,
}\left[  h_{j}^{l}\right]  _{j\in G_{l}}\right)  \right)  .
\]
It is expected that the limit%
\[
\frac{e^{-E(\boldsymbol{v},\boldsymbol{h})}}{Z^{\text{phys}}}d\boldsymbol{v}%
d\boldsymbol{h}\overset{\text{def}}{=}\lim_{l\rightarrow\infty}\boldsymbol{P}%
_{l}(\boldsymbol{v}_{l},\boldsymbol{h}_{l})\text{ }d^{\#G_{l}}\boldsymbol{v}%
\text{ }d^{\#G_{l}}\boldsymbol{h}%
\]
exists in some sense.

The marginal probability distributions are given by%
\begin{equation}
\boldsymbol{P}_{l}(\boldsymbol{v}_{l})=%
%TCIMACRO{\dsum \limits_{\boldsymbol{h}_{l}}}%
%BeginExpansion
{\displaystyle\sum\limits_{\boldsymbol{h}_{l}}}
%EndExpansion
\boldsymbol{P}_{l}(\boldsymbol{v}_{l},\boldsymbol{h}_{l})=\frac{%
%TCIMACRO{\dsum \limits_{\boldsymbol{h}_{l}}}%
%BeginExpansion
{\displaystyle\sum\limits_{\boldsymbol{h}_{l}}}
%EndExpansion
\exp\left(  -E_{l}\left(  \boldsymbol{v}_{l},\boldsymbol{h}_{l}\right)
\right)  }{%
%TCIMACRO{\dsum \limits_{\boldsymbol{v}_{l},\boldsymbol{h}_{l}}}%
%BeginExpansion
{\displaystyle\sum\limits_{\boldsymbol{v}_{l},\boldsymbol{h}_{l}}}
%EndExpansion
\exp\left(  -E_{l}\left(  \boldsymbol{v}_{l},\boldsymbol{h}_{l}\right)
\right)  }, \label{Probability_P_l_v}%
\end{equation}%
\[
\boldsymbol{P}_{l}(\boldsymbol{h}_{l})=%
%TCIMACRO{\dsum \limits_{\boldsymbol{v}_{l}}}%
%BeginExpansion
{\displaystyle\sum\limits_{\boldsymbol{v}_{l}}}
%EndExpansion
\boldsymbol{P}_{l}(\boldsymbol{v}_{l},\boldsymbol{h}_{l})=\frac{%
%TCIMACRO{\dsum \limits_{\boldsymbol{v}_{l}}}%
%BeginExpansion
{\displaystyle\sum\limits_{\boldsymbol{v}_{l}}}
%EndExpansion
\exp\left(  -E_{l}\left(  \boldsymbol{v}_{l},\boldsymbol{h}_{l}\right)
\right)  }{%
%TCIMACRO{\dsum \limits_{\boldsymbol{v}_{l},\boldsymbol{h}_{l}}}%
%BeginExpansion
{\displaystyle\sum\limits_{\boldsymbol{v}_{l},\boldsymbol{h}_{l}}}
%EndExpansion
\exp\left(  -E_{l}\left(  \boldsymbol{v}_{l},\boldsymbol{h}_{l}\right)
\right)  }.
\]

\subsubsection{Tree-like structures, $p$-adic numbers and DBNs}

The restriction of $\left\vert \cdot\right\vert _{p}$ to $G_{l}$ induces an
absolute value and $\left\vert G_{l}\right\vert _{p}=\left\{  0,p^{-\left(
l-1\right)  },\cdots,p^{-1},1\right\}  $. We endow $G_{l}$ with the metric
induced by $\left\vert \cdot\right\vert _{p}$, and thus $G_{l}$ becomes a
finite ultrametric space. In addition, $G_{l}$ can be identified with the set
of branches (vertices at the top level) of a rooted tree with $l+1$ levels and
$p^{l}$ branches. By definition the root of the tree is the only vertex at
level $0$. There are exactly $p$ vertices at level $1$, which correspond with
the possible values of the digit $i_{0}$ in the $p$-adic expansion of
$i\boldsymbol{=}i_{0}+i_{1}p+\ldots+i_{l-1}p^{l-1}$. Each of these vertices is
connected to the root by a non-directed edge. At level $k$, with $1<k\leq
l+1$, there are exactly $p^{k}$ vertices, \ each vertex corresponds to a
truncated expansion of $i$ of the form $i_{0}+\cdots+i_{k-1}p^{k-1}$. The
vertex corresponding to $i_{0}+\cdots+i_{k-1}p^{k-1}$ is connected to a vertex
$i_{0}^{\prime}+\cdots+i_{k-2}^{\prime}p^{k-2}$ at the level $k-1$ if and only
if $\left(  i_{0}+\cdots+i_{k-1}p^{k-1}\right)  -\left(  i_{0}^{\prime}%
+\cdots+i_{k-2}^{\prime}p^{k-2}\right)  $ is divisible by $p^{k-1}$. The unit
ball $\mathbb{Z}_{p}$ is \ an infinite rooted tree.

\begin{figure}[ptb]
\begin{center}
\includegraphics[height=1.8585in,width=4.5161in,angle=0]{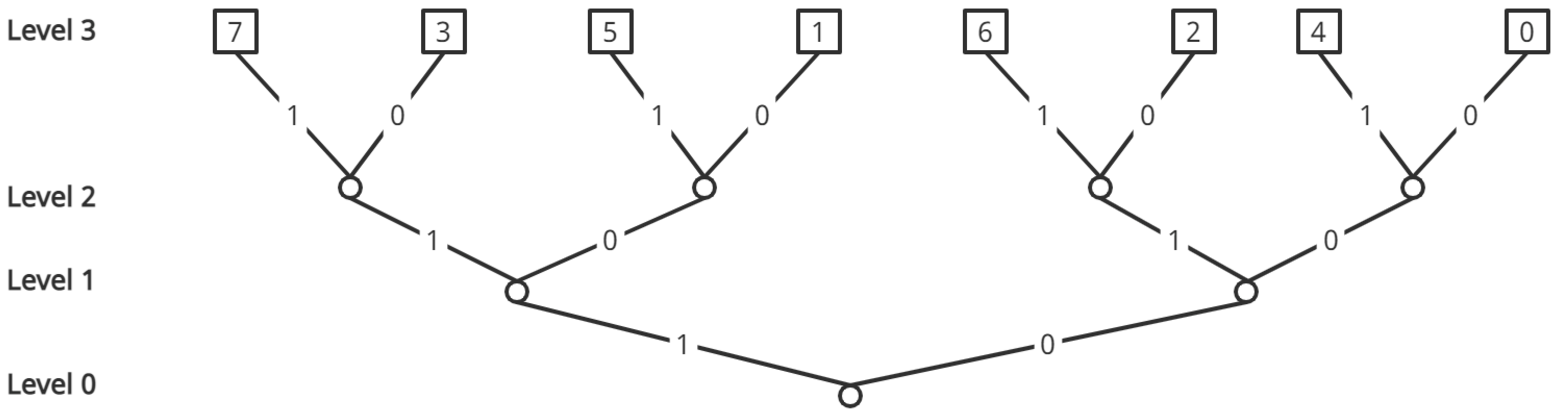}
\end{center}
\caption{{\small \textsl{The rooted tree associated with the group
$\mathbb{Z}_{2}/2^{3}\mathbb{Z}_{2}$. We identify the elements of
$\mathbb{Z}_{2}/2^{3}\mathbb{Z}_{2}$ with the set of integers $\left\{
0,\ldots,7\right\}  $ \ with binary representation $\boldsymbol{i}%
=\boldsymbol{i}_{0}+\boldsymbol{i}_{1}2+\boldsymbol{i}_{3}2^{2}%
,\;\;\;\boldsymbol{i}_{0},\boldsymbol{i}_{1},\boldsymbol{i}_{2}\in\{0,1\}$.
Two leaves $\boldsymbol{i},\boldsymbol{j}\in\mathbb{Z}_{2}/2^{3}\mathbb{Z}%
_{2}$ have a common ancestor at level $2$ if and only if $\boldsymbol{i}%
\equiv\boldsymbol{j}$ $\operatorname{mod}$ $2^{2}$, i.e., $\boldsymbol{i}%
=\boldsymbol{a}_{0}+\boldsymbol{a}_{1}2+\boldsymbol{i}_{2}2^{2}$ and
$\boldsymbol{j}=\boldsymbol{a}_{0}+\boldsymbol{a}_{1}2+\boldsymbol{j}_{2}%
2^{2}$ with $\boldsymbol{i}_{2},\boldsymbol{j}_{2}\in\{0,1\}$. Now, for
$\boldsymbol{i},\boldsymbol{j\in}\mathbb{Z}_{2}/2^{3}\mathbb{Z}_{2}$ have a
common ancestor at level $1$ if and only if $\boldsymbol{i}\equiv
\boldsymbol{j}$ $\operatorname{mod}$ $2$. Notice that that the $p$-adic
distance satisfies $-\log_{2}\left\vert \boldsymbol{i}-\boldsymbol{j}%
\right\vert _{2}=-\left(  \text{level of the first common ancestor of
}\boldsymbol{i}\text{, }\boldsymbol{j}\right)  $. Reprinted from
\cite{Zambrano-Zuniga}. }}}%
\label{fig:Figure 1}%
\end{figure}

We denote by $DBN(p,l,\boldsymbol{\theta}_{l})$ the $p$-adic discrete DBN with
energy functional $E_{l}\left(  \boldsymbol{v}_{l},\boldsymbol{h}%
_{l};\boldsymbol{\theta}_{l}\right)  $, see (\ref{Energy_lI}) and marginal
distribution $\boldsymbol{P}_{l}(\boldsymbol{v}_{l};\boldsymbol{\theta}_{l})$,
see (\ref{Probability_P_l_v}).

We now identify $G_{l}$ with the set of branches (vertices at the top level)
of a rooted tree with $l+1$ levels and $p^{l}$ branches. Attached to each
branch $i\in G_{l}$ there are two states : $v_{i}^{l}$, $h_{i}^{l}$.
The\textit{ visible field} is $\boldsymbol{v}_{l}=\left[  v_{i}^{l}\right]
_{i\in G_{l}}\in\left\{  0,1\right\}  ^{\#G_{l}}$ and \textit{the hidden field
is }$\boldsymbol{h}_{l}=\left[  h_{i}^{l}\right]  _{i\in G_{l}}\in\left\{
0,1\right\}  ^{\#G_{l}}$. These values are realizations of two random vectors,
that we also called the visible and hidden fields. The values $v_{i}^{l}$,
$v_{j}^{l}$ , respectively $h_{i}^{l}$, $h_{j}^{l}$, are statistically
independent for $i\neq j$. The $DBN(p,l,\boldsymbol{\theta}_{l})$ is a
$p-$adic analogue of the convolutional deep belief networks studied in
\cite{Honglak et al}. However, there are several important differences. We
discuss these matters in the last section of this article. We denote by
$DBN(p,\infty,\boldsymbol{\theta})$ the $p$-adic deep belief network
associated with the energy function $E\left(  \boldsymbol{v},\boldsymbol{h}%
\right)  $, see (\ref{Energy_Functional}). The $DBN(p,l,\boldsymbol{\theta
}_{l})$ for $l\geq L$ is a discretization of $DBN(p,\infty,\boldsymbol{\theta
})$, and \ for $l^{\prime}>l\geq L$, $DBN(p,l^{\prime},\boldsymbol{\theta
}_{l^{\prime}})$ is a larger scaled version of $DBN(p,l,\boldsymbol{\theta
}_{l})$. In general, the action $E_{l}\left(  \boldsymbol{v}_{l}%
,\boldsymbol{h}_{l};\boldsymbol{\theta}_{l}\right)  $ is non local, which
means that $w_{i-j}^{l}\neq0$ for any $i$, $j\in G_{l}$.

\section{\label{Section3}The $p$-adic DBNs are universal approximators}

We denote by$\boldsymbol{\ Q}(\boldsymbol{v})$ an arbitrary probability
distribution on a finite set
\[
\left\{  \boldsymbol{v}_{1},\ldots,\boldsymbol{v}_{2^{m}}\right\}  =\left\{
0,1\right\}  ^{m}%
\]
with $2^{m}$ elements. We fix $p$, a prime number, and $l_{0}$ a positive
integer such that $2^{m}\leq p^{l_{0}}$, and extend $\boldsymbol{Q}%
(\boldsymbol{v})$ to the set $\left\{  \boldsymbol{v}_{1},\ldots
,\boldsymbol{v}_{2^{m}},\boldsymbol{v}_{2^{m}+1},\ldots,\boldsymbol{v}%
_{p^{l_{0}}}\right\}  $ by taking $\boldsymbol{Q}(\boldsymbol{v}_{k})=0$ for
$2^{m}+1\leq k\leq p^{l_{0}}$. This observation allows us to extend
$\boldsymbol{Q}(\boldsymbol{v})$ to any finite set with cardinality $p^{l}$,
for any $l\geq l_{0}$. By identifying $\boldsymbol{v}$ with $\boldsymbol{v}%
_{l}=\left[  v_{j}^{l}\right]  _{j\in G_{l}}$, we can interpret
$\boldsymbol{Q}(\boldsymbol{v}_{l})$ a probability distributions on the
$\boldsymbol{v}_{l}$s.

In this section, we consider the problem of approximating $\boldsymbol{Q}%
(\boldsymbol{v})$ by the marginal distribution $\boldsymbol{P}_{l}%
(\boldsymbol{v})$ of an $DBN(p,l,\boldsymbol{\theta})$. To measure the
\textquotedblleft distance\textquotedblright\ between $\boldsymbol{Q}%
(\boldsymbol{v})$\ and $\boldsymbol{P}_{l}(\boldsymbol{v})$ we use the
Kullback-Leibler (KL) divergence:%
\[
KL(\boldsymbol{Q}\mid\boldsymbol{P}_{l})=%
%TCIMACRO{\dsum \limits_{\boldsymbol{v}}}%
%BeginExpansion
{\displaystyle\sum\limits_{\boldsymbol{v}}}
%EndExpansion
\boldsymbol{Q}(\boldsymbol{v})\ln\frac{\boldsymbol{Q}(\boldsymbol{v}%
)}{\boldsymbol{P}_{l}(\boldsymbol{v})}=-H(\boldsymbol{Q})-\frac{1}{\#G_{l}}%
%TCIMACRO{\dsum \limits_{j\in G_{l}}}%
%BeginExpansion
{\displaystyle\sum\limits_{j\in G_{l}}}
%EndExpansion
\ln\boldsymbol{P}_{l}(\boldsymbol{v}_{j}),
\]
where $H(\boldsymbol{Q})$ is the entropy of $\boldsymbol{Q}$. We recall that
$KL(\boldsymbol{Q}\mid\boldsymbol{P}_{l})=0$ if and only if $\boldsymbol{Q}%
=\boldsymbol{P}_{l}$. We construct an improved version of
$DBN(p,l,\boldsymbol{\theta}_{l})$ by increasing the the number of levels (or
layers) $l$, and consequently, the number of hidden variables (units), but
keeping the number of visible variables fixed.

\subsection{\label{SEction-Key-Construction}The key construction}

\begin{remark}
Given a positive integer $m$, and $a$, $b$ integers, we write%
\[
a\equiv b\text{ }\operatorname{mod}m
\]
to mean that $m$ divides $a-b$ in $\mathbb{Z}$, i.e., there exists
$n\in\mathbb{Z}$ such that $a-b=nm$.
\end{remark}

We first recall that $G_{l}=\mathbb{Z}_{p}/p^{l}\mathbb{Z}_{p}$ is isomorphic
to $\mathbb{Z}/p^{l}\mathbb{Z}$ as Abelian groups. We identify $i_{0}%
+i_{1}p+\ldots+i_{l-1}p^{l-1}\in G_{l}$ with an integer written in base $p$,
and the addition in $G_{l}$ with the sum of integers $\operatorname{mod}$
$p^{l}$. There is a natural homomorphism of Abelian groups:%
\[%
\begin{array}
[c]{ccc}%
G_{l+1} & \rightarrow & G_{l}\\
i & \rightarrow & i\text{ }\operatorname{mod}p^{l}.
\end{array}
\]
But, there are no natural homomorphisms from $G_{l}$ into $G_{l+1}$.

We identify $G_{l}$ with the \textit{subset} of $G_{l+1}$consisting of
integers having the form $i_{0}+i_{1}p+\ldots+i_{l-1}p^{l-1}$, where the
$i_{k}$s are $p$-adic digits. However, with this identification $G_{l}$ is not
a subgroup of the additive group $G_{l+1}$, because $G_{l}$ is not closed
under the addition in $G_{l+1}$. Indeed, $ap^{l-1}$, $\left(  p-a\right)
p^{l-1}\in G_{l}$ for any $a\in\left\{  1,\ldots,p-1\right\}  $, but
$ap^{l-1}+\left(  p-a\right)  p^{l-1}=p^{l}\notin G_{l}$.

We set
\[
T_{l+1}^{\ast}=\left\{  ap^{l};a\in\left\{  1,\ldots,p-1\right\}  \right\}
\subset G_{l+1},
\]
and $T_{l+1}=T_{l+1}^{\ast}\cup\left\{  0\right\}  $. Then $T_{l+1}$ is an
additive subgroup\ of $G_{l+1}$. Furthermore, as sets, it verifies that%
\[
G_{l+1}=%
%TCIMACRO{\dbigsqcup \limits_{k\in T_{l+1}}}%
%BeginExpansion
{\displaystyle\bigsqcup\limits_{k\in T_{l+1}}}
%EndExpansion
\left(  G_{l}+k\right)  ,
\]
where $%
%TCIMACRO{\tbigsqcup }%
%BeginExpansion
{\textstyle\bigsqcup}
%EndExpansion
$\ denotes the disjoint union and "$+$" denotes the addition in the group
$\left(  G_{l+1},+\right)  $.

We now construct a copy $\boldsymbol{v}_{l+1}=\left[  v_{i}^{l+1}\right]
_{i\in G_{l+1}}$ in $G_{l+1}$ of the visible field $\boldsymbol{v}_{l}=\left[
v_{i}^{l}\right]  _{i\in G_{l}}$. We set%
\[
\boldsymbol{v}_{j}^{l+1}=\boldsymbol{v}_{k}^{l}\text{ where }j\equiv k\text{
}\operatorname{mod}p^{l}\text{, for }j\in G_{l+1}\text{, }k\in G_{l}\text{.}%
\]
This construction is illustrated in Figure 2. 

\begin{figure}[ptb]
\begin{center}
\includegraphics[height=2.2693in,width=2.5642in,angle=0]{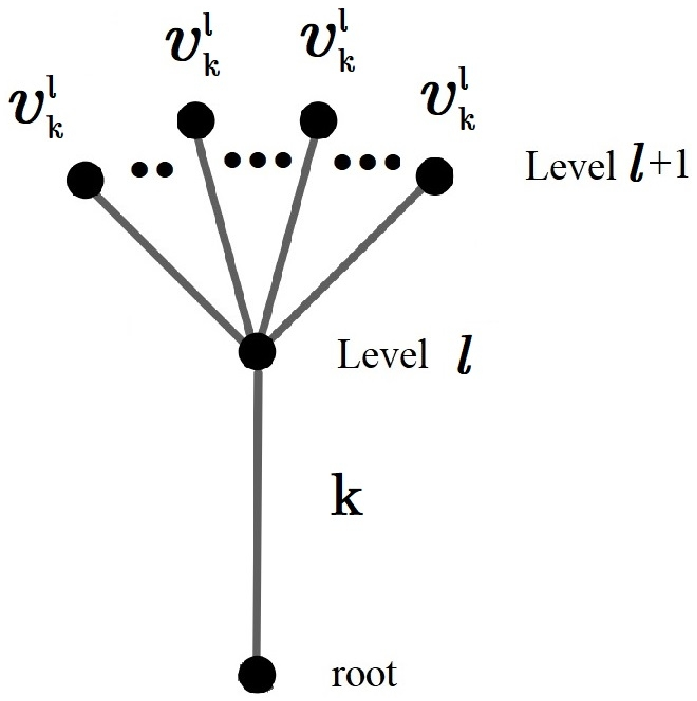}
\end{center}
\caption{Construction of the field $\boldsymbol{v}_{l+1}\boldsymbol{=}\left[
\boldsymbol{v}_{\boldsymbol{k}}^{l+1}\right]  _{\boldsymbol{k}\in G_{l+1}}$.}%
\label{fig:Figura 2}%
\end{figure}

We fix $j_{0}=\alpha p^{l}\in T_{l+1}^{\ast}$, $\alpha\in\left\{
1,\ldots,p-1\right\}  $, and set
\[
h_{j}^{l+1}=\left\{
\begin{array}
[c]{lll}%
0 & \text{if} & j\in G_{l+1}\smallsetminus\left(  G_{l}%
%TCIMACRO{\tbigsqcup }%
%BeginExpansion
{\textstyle\bigsqcup}
%EndExpansion
\left\{  j_{0}\right\}  \right)  \\
&  & \\
h_{j_{0}}^{l+1} & \text{if} & j=j_{0}\in T_{l+1}^{\ast}\\
&  & \\
h_{j}^{l} & \text{if } & j\in G_{l}.
\end{array}
\right.
\]
With this construction the hidden field $\boldsymbol{h}_{l+1}=\left[
h_{i}^{l+1}\right]  _{i\in G_{l+1}}$ of the new $DBN(p,l+1,\boldsymbol{\theta
}_{l+1})$ consists of the hidden field $\boldsymbol{h}_{l}=\left[  h_{i}%
^{l}\right]  _{i\in G_{l}}$ of $DBN(p,l,\boldsymbol{\theta}_{l})$ plus an
extra hidden unit $h_{j_{0}}^{l+1}$. This construction is illustrated in
Figure 3.
\begin{figure}[ptb]
\begin{center}
\includegraphics[height=2.968in,width=5.0246in,angle=0]{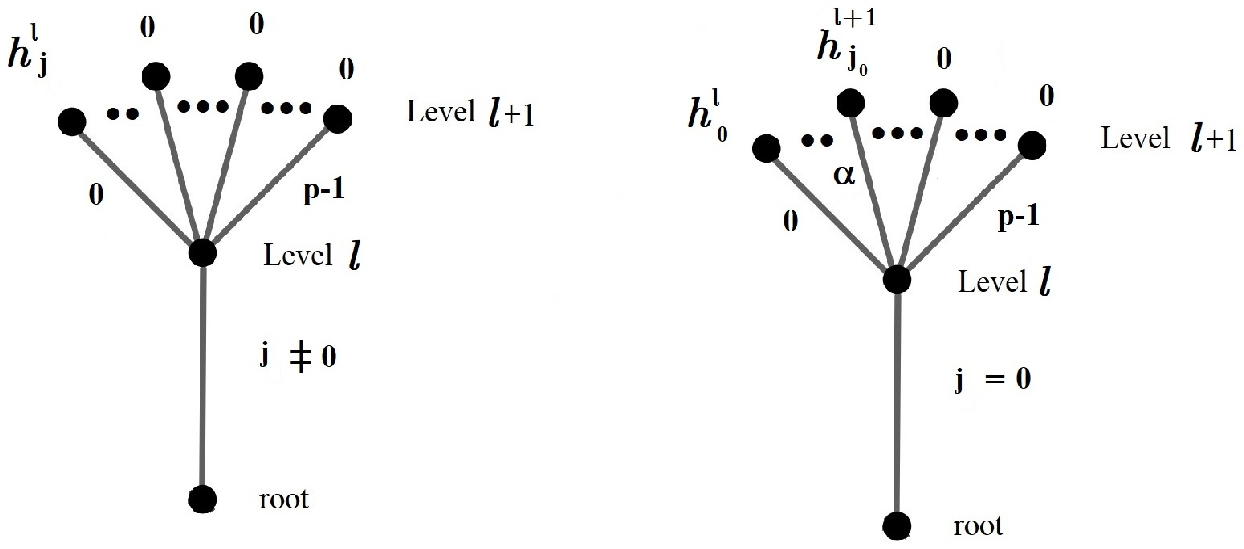}
\end{center}
\caption{Construction of the field $\boldsymbol{h}_{l+1}=\left[  h_{j}%
^{l+1}\right]  _{j\in G_{l+1}}$.}%
\label{fig:Figura 3}%
\end{figure}

We also set
\[
\text{\ }a_{j}^{l+1}=\left\{
\begin{array}
[c]{lll}%
0 & \text{if} & j\in G_{l+1}\smallsetminus G_{l}\\
&  & \\
a_{j}^{l} & \text{if } & j\in G_{l,}%
\end{array}
\right.  \text{ \ \ \ \ }b_{j}^{l+1}=\left\{
\begin{array}
[c]{lll}%
0 & \text{if} & j\in G_{l+1}\smallsetminus\left(  G_{l}%
%TCIMACRO{\tbigsqcup }%
%BeginExpansion
{\textstyle\bigsqcup}
%EndExpansion
T_{l+1}^{\ast}\right) \\
&  & \\
b_{j_{0}}^{l+1} & \text{if} & j=j_{0}\in T_{l+1}^{\ast}\\
&  & \\
0 & \text{if} & j\in T_{l+1}^{\ast}\smallsetminus\left\{  j_{0}\right\} \\
&  & \\
b_{j}^{l} & \text{if } & j\in G_{l,}%
\end{array}
\right.
\]
and $\boldsymbol{a}_{l+1}=\left[  a_{j}^{l+1}\right]  _{j\in G_{l+1}}$,
$\boldsymbol{b}_{l+1}=\left[  b_{j}^{l+1}\right]  _{j\in G_{l+1}}$.

We fix $\beta\in\left\{  1,\ldots,p-1\right\}  $ and construct a copy
$-G_{l}+\beta p^{l}$ of $G_{l}$ in $G_{l+1}$. Here the minus sign denotes the
inverse in the group $G_{l}$, and thus the set $-G_{l}$ is a permutation of
$G_{l}$. We now set%
\[
w_{j}^{l+1}=\left\{
\begin{array}
[c]{lll}%
0 & \text{if} & j\in G_{l+1}\smallsetminus\left(  G_{l}%
%TCIMACRO{\tbigsqcup }%
%BeginExpansion
{\textstyle\bigsqcup}
%EndExpansion
\left(  -G_{l}+\beta p^{l}\right)  \right) \\
&  & \\
w_{j}^{l+1} & \text{if} & j\in-G_{l}+\beta p^{l}\\
&  & \\
w_{j}^{l} & \text{if } & j\in G_{l}.
\end{array}
\right.
\]
Then, the vector
\[
\boldsymbol{w}_{l+1}=\left[  w_{j}^{l+1}\right]  _{j\in G_{l+1}}=\left[
\begin{array}
[c]{c}%
\boldsymbol{w}_{l}\\
\boldsymbol{w}_{l+1}^{\ast}%
\end{array}
\right]
\]
consists\ of the weigh vector $\boldsymbol{w}_{l}$ of
$DBN(p,l,\boldsymbol{\theta}_{l})$ and a new vector $\boldsymbol{w}%
_{l+1}^{\ast}=\left[  w_{j}^{l+1}\right]  _{j\in-G_{l}+\beta p^{l}}$.

Given an $DBN(p,l,\boldsymbol{\theta}_{l})$, the key construction \ allows us
to construct a new $DBN(p,l+1,\boldsymbol{\theta}_{l+1})$, $\boldsymbol{\theta
}_{l+1}=\left(  \boldsymbol{\theta}_{l},\boldsymbol{w}_{l+1}^{\ast},b_{j_{0}%
}^{l+1}\right)  $, with an extra layer, and extra hidden unit $h_{j_{0}}%
^{l+1}\in\left\{  0,1\right\}  $, and $1+\#G_{l}$\ new parameters: $b_{j_{0}%
}^{l+1}\in\mathbb{R}$, $\boldsymbol{w}_{l+1}^{\ast}\in\mathbb{R}^{\#G_{l}}$.
The energy functional of the new $DBN(p,l+1,\boldsymbol{\theta}_{l+1})$ is
given by%
\begin{align*}
E_{l+1}\left(  \boldsymbol{v}_{l+1},\boldsymbol{h}_{l+1},\boldsymbol{\theta
}_{l+1}\right)   &  =%
%TCIMACRO{\dsum \limits_{j\in G_{l+1}}}%
%BeginExpansion
{\displaystyle\sum\limits_{j\in G_{l+1}}}
%EndExpansion
\text{ }%
%TCIMACRO{\dsum \limits_{k\in G_{l+1}}}%
%BeginExpansion
{\displaystyle\sum\limits_{k\in G_{l+1}}}
%EndExpansion
w_{k}^{l+1}v_{j}^{l+1}h_{j+k}^{l+1}+%
%TCIMACRO{\dsum \limits_{j\in G_{l+1}}}%
%BeginExpansion
{\displaystyle\sum\limits_{j\in G_{l+1}}}
%EndExpansion
a_{j}^{l+1}v_{j}^{l+1}\\
&  +%
%TCIMACRO{\dsum \limits_{j\in G_{l+1}}}%
%BeginExpansion
{\displaystyle\sum\limits_{j\in G_{l+1}}}
%EndExpansion
b_{j}^{l+1}h_{j}^{l+1},
\end{align*}
where $\boldsymbol{\theta}_{l+1}=\left(  \boldsymbol{w}_{l+1},\boldsymbol{a}%
_{l+1},\boldsymbol{b}_{l+1}\right)  $.

\begin{lemma}
\label{Lemma_1}With the above notation, the following formulas holds true. For
$a$, $b\in\left\{  0,\ldots,p-1\right\}  $, we set
\[
S(a,b)=%
%TCIMACRO{\dsum \limits_{j\in G_{l}+ap^{l}}}%
%BeginExpansion
{\displaystyle\sum\limits_{j\in G_{l}+ap^{l}}}
%EndExpansion
\text{ }%
%TCIMACRO{\dsum \limits_{k\in G_{l}+bp^{l}}}%
%BeginExpansion
{\displaystyle\sum\limits_{k\in G_{l}+bp^{l}}}
%EndExpansion
w_{k}^{l+1}v_{j}^{l+1}h_{j+k}^{l+1},
\]
where \textquotedblleft$+$\textquotedblright\ denotes the sum in $G_{l+1}$.

\noindent(i) If $\ a+b\not \equiv 0$ $\operatorname{mod}p$, then%
\[
S(a,b)=\left\{
\begin{array}
[c]{ll}%
h_{j_{0}}^{l+1}%
%TCIMACRO{\dsum \limits_{j\in G_{l}}}%
%BeginExpansion
{\displaystyle\sum\limits_{j\in G_{l}}}
%EndExpansion
w_{-j+\beta p^{l}}^{l+1}v_{j}^{l} & \text{if }\left(  a+b\right)
p^{l}+j+k=\alpha p^{l}\text{, for some }j,k\in G_{l}\\
& \\
0 & \text{otherwise.}%
\end{array}
\right.
\]

\noindent(ii) If $\ a+b\not \equiv 0$ $\operatorname{mod}p$, $\ $ then
\[%
%TCIMACRO{\dsum \limits_{a+b\not \equiv 0\text{ }\operatorname{mod}p}}%
%BeginExpansion
{\displaystyle\sum\limits_{a+b\not \equiv 0\text{ }\operatorname{mod}p}}
%EndExpansion
S(a,b)=ph_{j_{0}}^{l+1}%
%TCIMACRO{\dsum \limits_{j\in G_{l}}}%
%BeginExpansion
{\displaystyle\sum\limits_{j\in G_{l}}}
%EndExpansion
w_{-j+\beta p^{l}}^{l+1}v_{j}^{l}.
\]

\noindent(iii) If $a+b\equiv0$ $\operatorname{mod}p$, then%
\[
S(a,b)=\left\{
\begin{array}
[c]{ccc}%
%TCIMACRO{\dsum \limits_{j\in G_{l}}}%
%BeginExpansion
{\displaystyle\sum\limits_{j\in G_{l}}}
%EndExpansion
\text{ }%
%TCIMACRO{\dsum \limits_{k\in G_{l}}}%
%BeginExpansion
{\displaystyle\sum\limits_{k\in G_{l}}}
%EndExpansion
w_{k}^{l}v_{j}^{l}h_{j+k}^{l} &  & \text{if }a\equiv0\operatorname{mod}%
p,b\equiv0\operatorname{mod}p\\
&  & \\
0 &  & \text{otherwise.}%
\end{array}
\right.
\]

\noindent(iv)
\begin{align*}
E_{l+1}\left(  \boldsymbol{v}_{l+1},\boldsymbol{h}_{l+1};\boldsymbol{\theta
}_{l+1}\right)   &  =E_{l+1}\left(  \boldsymbol{v}_{l},\boldsymbol{h}%
_{l};\boldsymbol{\theta}_{l+1}\right) \\
&  =E_{l}\left(  \boldsymbol{v}_{l},\boldsymbol{h}_{l};\boldsymbol{\theta}%
_{l}\right)  +ph_{j_{0}}^{l+1}%
%TCIMACRO{\dsum \limits_{j\in G_{l}}}%
%BeginExpansion
{\displaystyle\sum\limits_{j\in G_{l}}}
%EndExpansion
w_{-j+\beta p^{l}}^{l+1}v_{j}^{l}+b_{j_{0}}^{l+1}h_{j_{0}}^{l+1}.
\end{align*}

\end{lemma}

\begin{proof}
(i) Notice that
\[
S(a,b)=%
%TCIMACRO{\dsum \limits_{j\in G_{l}}}%
%BeginExpansion
{\displaystyle\sum\limits_{j\in G_{l}}}
%EndExpansion
\text{ }%
%TCIMACRO{\dsum \limits_{k\in G_{l}}}%
%BeginExpansion
{\displaystyle\sum\limits_{k\in G_{l}}}
%EndExpansion
w_{k+bp^{l}}^{l+1}v_{j+ap^{l}}^{l+1}h_{\left(  a+b\right)  p^{l}+j+k}^{l+1}.
\]
We now use that $h_{\left(  a+b\right)  p^{l}+j+k}^{l+1}\neq0$ if and only if
$\left(  a+b\right)  p^{l}+j+k=\alpha p^{l}$ in $G_{l+1}$, i.e., if and only
if
\begin{equation}
\left(  a+b\right)  p^{l}+j+k\equiv\alpha p^{l}\text{ }\operatorname{mod}%
p^{l+1}. \label{Condition_1}%
\end{equation}
The condition (\ref{Condition_1}) implies that%
\[
j+k\equiv0\text{ }\operatorname{mod}p^{s}\text{ for }1\leq s\leq l,
\]
which in turn implies that $j+k=0$ in $G_{l}$, and thus (\ref{Condition_1})
becomes%
\begin{equation}
a+b\equiv\alpha\text{ }\operatorname{mod}p. \label{Condition_1A}%
\end{equation}
This last congruence has solutions since $a+b\not \equiv 0$
$\operatorname{mod}p$ . If (\ref{Condition_1}) is satisfied, then $k=-j$ in
$G_{l}$ and since $w_{-j+bp^{l}}^{l+1}\neq0\Leftrightarrow b=\beta$, and
$v_{j+ap^{l}}^{l+1}=v_{j}^{l}$, we have
\[
S(a,b)=h_{j_{0}}^{l+1}%
%TCIMACRO{\dsum \limits_{j\in G_{l}}}%
%BeginExpansion
{\displaystyle\sum\limits_{j\in G_{l}}}
%EndExpansion
w_{-j+\beta p^{l}}^{l+1}v_{j}^{l},
\]
otherwise $S(a,b)=0$.

(ii) It follows form the first part by using that%
\begin{multline*}%
%TCIMACRO{\dsum \limits_{a+b\not \equiv 0\operatorname{mod}p}}%
%BeginExpansion
{\displaystyle\sum\limits_{a+b\not \equiv 0\operatorname{mod}p}}
%EndExpansion
S(a,b)=%
%TCIMACRO{\dsum \limits_{a\not \equiv 0\operatorname{mod}p\ }}%
%BeginExpansion
{\displaystyle\sum\limits_{a\not \equiv 0\operatorname{mod}p\ }}
%EndExpansion%
%TCIMACRO{\dsum \limits_{\text{ }b\not \equiv p-a\operatorname{mod}p}}%
%BeginExpansion
{\displaystyle\sum\limits_{\text{ }b\not \equiv p-a\operatorname{mod}p}}
%EndExpansion
S(a,b)+%
%TCIMACRO{\dsum \limits_{a\equiv0\operatorname{mod}p\ }}%
%BeginExpansion
{\displaystyle\sum\limits_{a\equiv0\operatorname{mod}p\ }}
%EndExpansion%
%TCIMACRO{\dsum \limits_{\text{ }b\not \equiv 0\operatorname{mod}p}}%
%BeginExpansion
{\displaystyle\sum\limits_{\text{ }b\not \equiv 0\operatorname{mod}p}}
%EndExpansion
S(a,b)\\
+%
%TCIMACRO{\dsum \limits_{a\not \equiv 0\operatorname{mod}p\ }}%
%BeginExpansion
{\displaystyle\sum\limits_{a\not \equiv 0\operatorname{mod}p\ }}
%EndExpansion%
%TCIMACRO{\dsum \limits_{\text{ }b\equiv0\operatorname{mod}p}}%
%BeginExpansion
{\displaystyle\sum\limits_{\text{ }b\equiv0\operatorname{mod}p}}
%EndExpansion
S(a,b).
\end{multline*}

(iii)\ Notice that $b\equiv\left(  p-a\right)  $ $\operatorname{mod}p$, then
\[
S(a,b)=%
%TCIMACRO{\dsum \limits_{j\in G_{l}}}%
%BeginExpansion
{\displaystyle\sum\limits_{j\in G_{l}}}
%EndExpansion
\text{ }%
%TCIMACRO{\dsum \limits_{k\in G_{l}}}%
%BeginExpansion
{\displaystyle\sum\limits_{k\in G_{l}}}
%EndExpansion
w_{k+\left(  p-a\right)  p^{l}}^{l+1}v_{j+ap^{l}}^{l+1}h_{j+k}^{l+1}.
\]
If $a\not \equiv 0$ $\operatorname{mod}p$, then $v_{j+ap^{l}}^{l+1}=0$, and
$S(a,b)=0$. If $a\equiv0$ $\operatorname{mod}p$, then $b\equiv0$
$\operatorname{mod}p$, and\
\begin{equation}
S(a,b)=%
%TCIMACRO{\dsum \limits_{j\in G_{l}}}%
%BeginExpansion
{\displaystyle\sum\limits_{j\in G_{l}}}
%EndExpansion%
%TCIMACRO{\dsum \limits_{k\in G_{l}}}%
%BeginExpansion
{\displaystyle\sum\limits_{k\in G_{l}}}
%EndExpansion
w_{k}^{l}v_{j}^{l}h_{j+k}^{l+1}=%
%TCIMACRO{\dsum \limits_{\substack{j,k\in G_{l}\\j+k\in G_{l}}}}%
%BeginExpansion
{\displaystyle\sum\limits_{\substack{j,k\in G_{l}\\j+k\in G_{l}}}}
%EndExpansion
w_{k}^{l}v_{j}^{l}h_{j+k}^{l+1}+%
%TCIMACRO{\dsum \limits_{\substack{j,k\in G_{l}\\j+k=j_{0}}}}%
%BeginExpansion
{\displaystyle\sum\limits_{\substack{j,k\in G_{l}\\j+k=j_{0}}}}
%EndExpansion
w_{k}^{l}v_{j}^{l}h_{j+k}^{l+1}. \label{Last_sum}%
\end{equation}
A simple inductive argument on $l$ shows that $j+k\equiv0$ $\operatorname{mod}%
p^{l}$, with $j,k\in G_{l}$, is only possible if $j=s+ap^{l-1}\in G_{l}$ and
$k=-s+\left(  p-a\right)  p^{l-1}\in G_{l}$, in this case, $j+k\boldsymbol{=}%
p^{l+1}\neq j_{0}$ in $G_{l+1}$ and consequently the last sum in
(\ref{Last_sum}) is zero.

(iv) We first notice that%
\begin{gather}
E_{l+1}\left(  \boldsymbol{v}_{l+1},\boldsymbol{h}_{l+1};\boldsymbol{\theta
}_{l+1}\right)  =%
%TCIMACRO{\dsum \limits_{j,k\in G_{l+1}}}%
%BeginExpansion
{\displaystyle\sum\limits_{j,k\in G_{l+1}}}
%EndExpansion
w_{k}^{l+1}v_{j+k}^{l+1}h_{j}^{l+1}+%
%TCIMACRO{\dsum \limits_{j\in G_{l}{{\bigsqcup}}T_{l+1}}}%
%BeginExpansion
{\displaystyle\sum\limits_{j\in G_{l}{{\bigsqcup}}T_{l+1}}}
%EndExpansion
a_{j}^{l+1}v_{j}^{l+1}\label{Formula_1}\\
+%
%TCIMACRO{\dsum \limits_{j\in G_{l}{\bigsqcup}T_{l+1}}}%
%BeginExpansion
{\displaystyle\sum\limits_{j\in G_{l}{\bigsqcup}T_{l+1}}}
%EndExpansion
b_{j}^{l+1}h_{j}^{l+1}\nonumber\\
=:S_{l+1}^{\left(  0\right)  }\left(  \boldsymbol{v}_{l+1},\boldsymbol{h}%
_{l+1},\boldsymbol{w}_{l+1}\right)  +S_{l+1}^{\left(  1\right)  }\left(
\boldsymbol{v}_{l+1},\boldsymbol{a}_{l+1}\right)  +S_{l+1}^{\left(  2\right)
}\left(  \boldsymbol{h}_{l+1},\boldsymbol{b}_{l+1}\right)  .\nonumber
\end{gather}
Now%
\begin{align}
S_{l+1}^{\left(  0\right)  }\left(  \boldsymbol{v}_{l+1},\boldsymbol{h}%
_{l+1};\boldsymbol{\theta}_{l+1}\right)   &  :=%
%TCIMACRO{\dsum \limits_{j\in G_{l+1}}}%
%BeginExpansion
{\displaystyle\sum\limits_{j\in G_{l+1}}}
%EndExpansion
\text{ }%
%TCIMACRO{\dsum \limits_{k\in G_{l+1}}}%
%BeginExpansion
{\displaystyle\sum\limits_{k\in G_{l+1}}}
%EndExpansion
w_{k}^{l+1}v_{j}^{l+1}h_{j+k}^{l+1}=%
%TCIMACRO{\dsum \limits_{a\ }}%
%BeginExpansion
{\displaystyle\sum\limits_{a\ }}
%EndExpansion%
%TCIMACRO{\dsum \limits_{\text{ }b}}%
%BeginExpansion
{\displaystyle\sum\limits_{\text{ }b}}
%EndExpansion
S(a,b)\label{Formula_2A}\\
&  =ph_{j_{0}}^{l+1}%
%TCIMACRO{\dsum \limits_{j\in G_{l}}}%
%BeginExpansion
{\displaystyle\sum\limits_{j\in G_{l}}}
%EndExpansion
w_{-j+\beta p^{l}}^{l+1}v_{j}^{l}+%
%TCIMACRO{\dsum \limits_{j\in G_{l}}}%
%BeginExpansion
{\displaystyle\sum\limits_{j\in G_{l}}}
%EndExpansion
\text{ }%
%TCIMACRO{\dsum \limits_{k\in G_{l}}}%
%BeginExpansion
{\displaystyle\sum\limits_{k\in G_{l}}}
%EndExpansion
w_{k}^{l}v_{j}^{l}h_{j+k}^{l}.\nonumber
\end{align}
It follows immediately that%
\begin{equation}
S_{l+1}^{\left(  1\right)  }\left(  \boldsymbol{v}_{l+1},\boldsymbol{a}%
_{l+1}\right)  =%
%TCIMACRO{\dsum \limits_{j\in G_{l}{\bigsqcup}T_{l+1}}}%
%BeginExpansion
{\displaystyle\sum\limits_{j\in G_{l}{\bigsqcup}T_{l+1}}}
%EndExpansion
a_{j}^{l+1}v_{j}^{l+1}=%
%TCIMACRO{\dsum \limits_{j\in G_{l}}}%
%BeginExpansion
{\displaystyle\sum\limits_{j\in G_{l}}}
%EndExpansion
a_{j}^{l}v_{j}^{l}, \label{Formula_2}%
\end{equation}
and that%
\begin{equation}
S_{l+1}^{\left(  2\right)  }\left(  \boldsymbol{h}_{l+1},\boldsymbol{b}%
_{l+1}\right)  =%
%TCIMACRO{\dsum \limits_{j\in G_{l}}}%
%BeginExpansion
{\displaystyle\sum\limits_{j\in G_{l}}}
%EndExpansion
b_{j}^{l}h_{j}^{l}+b_{j_{0}}^{l+1}h_{j_{0}}^{l+1}. \label{Formula_3}%
\end{equation}
The announce formula follows from formula (\ref{Formula_1}), by using
(\ref{Formula_2A})-(\ref{Formula_3}).
\end{proof}

It is relevant to mention that the energy functional of
$DBN(p,l+1,\boldsymbol{\theta}_{l+1})$ is an extension of the energy
functional of $DBN(p,l,\boldsymbol{\theta}_{l})$. Furthermore, the key
construction is recursive. Starting with $DBN(p,l+1,\boldsymbol{\theta}%
_{l+1})$ there exists another $DBN(p,l+2,\boldsymbol{\theta}_{l+2})$ whose
energy functional is an extension of the energy functional of
$DBN(p,l+1,\boldsymbol{\theta}_{l+1})$.

\subsection{Better model with increasing number of levels}

In this section we show that the computational power of an $DBN(p,l,\theta)$
increases with the number of levels (or layers). More precisely, we show
$p$-adic counterparts of the main results in \cite[Theorems 1, 2]{Le roux et
al 1}.

On the other hand, $\left\{  -j+\beta p^{l};j\in G_{l}\right\}  \subset
G_{l+1}$ \ is a copy (more precisely a fixed lifting) of $G_{l}$ in $G_{l+1}$,
and since the $w_{-j+\beta p^{l}}^{l+1}$s are new parameters, we rename
$w_{-j+\beta p^{l}}^{l+1}$ as $w_{j}^{l+1}$, then%
\[%
%TCIMACRO{\dsum \limits_{j\in G_{l}}}%
%BeginExpansion
{\displaystyle\sum\limits_{j\in G_{l}}}
%EndExpansion
w_{j}^{l+1}v_{j}^{l}=%
%TCIMACRO{\dsum \limits_{j\in G_{l}}}%
%BeginExpansion
{\displaystyle\sum\limits_{j\in G_{l}}}
%EndExpansion
w_{-j+\beta p^{l}}^{l+1}v_{j}^{l}.
\]
We rescale $h_{j_{0}}^{l+1}$ to $ph_{j_{0}}^{l+1}$ and $b_{j_{0}}^{l+1}$ to
$p^{-1}b_{j_{0}}^{l+1}$. With this notation the energy functional
$E_{l+1}\left(  \boldsymbol{v}_{l},\boldsymbol{h}_{l+1};\boldsymbol{\theta
}_{l+1}\right)  $ becomes
\begin{equation}
E_{l+1}\left(  \boldsymbol{v}_{l},\boldsymbol{h}_{l+1};\boldsymbol{\theta
}_{l+1}\right)  =E_{l}\left(  \boldsymbol{v}_{l},\boldsymbol{h}_{l}%
;\boldsymbol{\theta}_{l}\right)  +h_{j_{0}}^{l+1}%
%TCIMACRO{\dsum \limits_{k\in G_{l}}}%
%BeginExpansion
{\displaystyle\sum\limits_{k\in G_{l}}}
%EndExpansion
w_{k}^{l+1}v_{k}^{l}+b_{j_{0}}^{l+1}h_{j_{0}}^{l+1}, \label{Key_Formula}%
\end{equation}
where $\boldsymbol{h}_{l+1}=\left[
\begin{array}
[c]{c}%
\boldsymbol{h}_{l}\\
h_{\boldsymbol{j}_{0}}^{l+1}%
\end{array}
\right]  $, $\boldsymbol{\theta}_{l+1}=\left(  \boldsymbol{\theta}%
_{l},\boldsymbol{w}_{l+1},b_{j_{0}}^{l+1}\right)  $. Notice that
$DBN(p,l+1,\boldsymbol{\theta}_{l+1})=DBN(p,l+1,\boldsymbol{\theta}%
_{l},\boldsymbol{w}_{l+1},b_{j_{0}}^{l+1})$ has only one additional hidden
unit ($h_{j_{0}}^{l+1}$). The corresponding Boltzmann distribution is given by%
\[
\boldsymbol{P}_{l+1}\left(  \boldsymbol{v}_{l},\boldsymbol{h}_{l+1}%
;\boldsymbol{\theta}_{l+1}\right)  =\frac{\exp\left(  -E_{l+1}\left(
\boldsymbol{v}_{l},\boldsymbol{h}_{l+1};\boldsymbol{\theta}_{l+1}\right)
\right)  }{Z_{l+1}(\boldsymbol{\theta}_{l+1})},
\]
and the marginal distribution is given by
\[
\boldsymbol{P}_{l+1}\left(  \boldsymbol{v}_{l};\boldsymbol{\theta}%
_{l+1}\right)  =\frac{%
%TCIMACRO{\dsum \limits_{\boldsymbol{h}_{l+1}}}%
%BeginExpansion
{\displaystyle\sum\limits_{\boldsymbol{h}_{l+1}}}
%EndExpansion
\exp\left(  -E_{l+1}\left(  \boldsymbol{v}_{l},\boldsymbol{h}_{l+1}%
;\boldsymbol{\theta}_{l+1}\right)  \right)  }{Z_{l+1}(\theta_{l+1})}.
\]

\begin{lemma}
\label{Lemma_2}Let $\boldsymbol{P}_{l}\left(  \boldsymbol{v}_{l}%
;\boldsymbol{\theta}_{l}\right)  $ be a probability distribution over binary
vectors $\left\{  0,1\right\}  ^{\#G_{l}}$ obtained with an
$DBN(p,l,\boldsymbol{\theta}_{l})$, and let $\boldsymbol{P}_{l+1}\left(
\boldsymbol{v}_{l};\boldsymbol{\theta}_{l+1}\right)  =\boldsymbol{P}%
_{l+1}\left(  \boldsymbol{v}_{l};\boldsymbol{\theta}_{l},\boldsymbol{w}%
_{l+1},b_{j_{0}}^{l+1}\right)  $ be the marginal probability distribution
corresponding to $DBN(p,l+1,\boldsymbol{\theta}_{l},\boldsymbol{w}%
_{l+1},b_{j_{0}}^{l+1})$, which is obtained from $DBN(p,l,\boldsymbol{\theta
}_{l})$ by adding one level and one hidden unit. Then $\boldsymbol{P}%
_{l+1}\left(  \boldsymbol{v}_{l+1};\boldsymbol{\theta}_{l},\boldsymbol{w}%
_{l+1},b_{j_{0}}^{l+1}\right)  $ is a probability distribution over binary
vectors $\left\{  0,1\right\}  ^{\#G_{l}}$ for any $b_{j_{0}}^{l+1}\in\left[
-\infty,\infty\right)  $, and $\boldsymbol{P}_{l+1}\left(  \boldsymbol{v}%
_{l};\boldsymbol{\theta}_{l},\boldsymbol{w}_{l+1},-\infty\right)
=\boldsymbol{P}_{l}\left(  \boldsymbol{v}_{l};\boldsymbol{\theta}_{l}\right)
$.
\end{lemma}

\begin{proof}
By using the formula (\ref{Key_Formula}) and the fact that $\boldsymbol{h}%
_{l+1}=\left(  \boldsymbol{h}_{l},h_{j_{0}}^{l+1}\right)  $, we have%
\begin{gather*}%
%TCIMACRO{\dsum \limits_{\boldsymbol{h}_{l+1}}}%
%BeginExpansion
{\displaystyle\sum\limits_{\boldsymbol{h}_{l+1}}}
%EndExpansion
\exp\left(  -E_{l+1}\left(  \boldsymbol{v}_{l},\boldsymbol{h}_{l+1}%
;\boldsymbol{\theta}_{l+1}\right)  \right) \\
=%
%TCIMACRO{\dsum \limits_{\boldsymbol{h}_{l},h_{j_{0}}^{l+1}}}%
%BeginExpansion
{\displaystyle\sum\limits_{\boldsymbol{h}_{l},h_{j_{0}}^{l+1}}}
%EndExpansion
\exp\left(  -E_{l}\left(  \boldsymbol{v}_{l},\boldsymbol{h}_{l}%
;\boldsymbol{\theta}_{l}\right)  \right)  \exp\left(  h_{j_{0}}^{l+1}%
%TCIMACRO{\dsum \limits_{k\in G_{l}}}%
%BeginExpansion
{\displaystyle\sum\limits_{k\in G_{l}}}
%EndExpansion
w_{k}^{l+1}v_{k}^{l}+b_{j_{0}}^{l+1}h_{j_{0}}^{l+1}\right) \\
=\left\{
%TCIMACRO{\dsum \limits_{\boldsymbol{h}_{l}}}%
%BeginExpansion
{\displaystyle\sum\limits_{\boldsymbol{h}_{l}}}
%EndExpansion
\exp\left(  -E_{l}\left(  \boldsymbol{v}_{l},\boldsymbol{h}_{l}%
;\boldsymbol{\theta}_{l}\right)  \right)  \right\}  \left\{
%TCIMACRO{\dsum \limits_{h_{j_{0}}^{l+1}\in\left\{  0,1\right\}  }}%
%BeginExpansion
{\displaystyle\sum\limits_{h_{j_{0}}^{l+1}\in\left\{  0,1\right\}  }}
%EndExpansion
\exp\left(  h_{j_{0}}^{l+1}%
%TCIMACRO{\dsum \limits_{k\in G_{l}}}%
%BeginExpansion
{\displaystyle\sum\limits_{k\in G_{l}}}
%EndExpansion
w_{k}^{l+1}v_{k}^{l}+b_{j_{0}}^{l+1}h_{j_{0}}^{l+1}\right)  \right\} \\
=\left\{
%TCIMACRO{\dsum \limits_{\boldsymbol{h}_{l}}}%
%BeginExpansion
{\displaystyle\sum\limits_{\boldsymbol{h}_{l}}}
%EndExpansion
\exp\left(  -E_{l}\left(  \boldsymbol{v}_{l},\boldsymbol{h}_{l}%
;\boldsymbol{\theta}_{l}\right)  \right)  \right\}  \left\{  1+\exp\left(
%TCIMACRO{\dsum \limits_{k\in G_{l}}}%
%BeginExpansion
{\displaystyle\sum\limits_{k\in G_{l}}}
%EndExpansion
w_{k}^{l+1}v_{k}^{l}+b_{j_{0}}^{l+1}\right)  \right\}  .
\end{gather*}
Then%
\begin{align}
Z_{l+1}(\boldsymbol{\theta}_{l},\boldsymbol{w}_{l+1},b_{j_{0}}^{l+1})  &  =%
%TCIMACRO{\dsum \limits_{\boldsymbol{v}_{l+1},\boldsymbol{h}_{l+1}}}%
%BeginExpansion
{\displaystyle\sum\limits_{\boldsymbol{v}_{l+1},\boldsymbol{h}_{l+1}}}
%EndExpansion
\exp\left(  -E_{l+1}\left(  \boldsymbol{v}_{l+1},\boldsymbol{h}_{l+1}%
;\boldsymbol{\theta}_{l+1}\right)  \right) \label{Formula_10}\\
&  =%
%TCIMACRO{\dsum \limits_{\boldsymbol{v}_{l},\boldsymbol{h}_{l}}}%
%BeginExpansion
{\displaystyle\sum\limits_{\boldsymbol{v}_{l},\boldsymbol{h}_{l}}}
%EndExpansion
\left\{  1+\exp\left(
%TCIMACRO{\dsum \limits_{k\in G_{l}}}%
%BeginExpansion
{\displaystyle\sum\limits_{k\in G_{l}}}
%EndExpansion
w_{k}^{l+1}v_{k}^{l}+b_{j_{0}}^{l+1}\right)  \right\}  \exp\left(
-E_{l}\left(  \boldsymbol{v}_{l},\boldsymbol{h}_{l};\boldsymbol{\theta}%
_{l}\right)  \right)  ,\nonumber
\end{align}
and%
\begin{gather}
\boldsymbol{P}_{l+1}\left(  \boldsymbol{v}_{l};\boldsymbol{\theta}%
_{l},\boldsymbol{w}_{l+1},b_{j_{0}}^{l+1}\right)  =\label{Formula_11}\\
\frac{\left\{  1+\exp\left(
%TCIMACRO{\dsum \limits_{k\in G_{l}}}%
%BeginExpansion
{\displaystyle\sum\limits_{k\in G_{l}}}
%EndExpansion
w_{k}^{l+1}v_{k}^{l}+b_{j_{0}}^{l+1}\right)  \right\}
%TCIMACRO{\dsum \limits_{\boldsymbol{h}_{l}}}%
%BeginExpansion
{\displaystyle\sum\limits_{\boldsymbol{h}_{l}}}
%EndExpansion
\exp\left(  -E_{l}\left(  \boldsymbol{v}_{l},\boldsymbol{h}_{l}%
;\boldsymbol{\theta}_{l}\right)  \right)  }{Z_{l+1}(\boldsymbol{\theta}%
_{l},\boldsymbol{w}_{l+1},b_{j_{0}}^{l+1})}.\nonumber
\end{gather}
Thus $\boldsymbol{P}_{l+1}\left(  \boldsymbol{v}_{l};\theta_{l},w_{l+1}%
,b_{j_{0}}^{l+1}\right)  $ is a well-defined probability distribution for any
$b_{j_{0}}^{l+1}\in\left[  -\infty,\infty\right)  $, and $\boldsymbol{P}%
_{l+1}\left(  \boldsymbol{v}_{l};\theta_{l},-\infty\right)  =\boldsymbol{P}%
_{l}\left(  \boldsymbol{v}_{l};\theta_{l}\right)  $.
\end{proof}

\begin{lemma}
\label{Lemma_3}Assume that $KL(\boldsymbol{Q}\left(  \boldsymbol{v}%
_{l}\right)  \mid\boldsymbol{P}_{l}\left(  \boldsymbol{v}_{l};\theta
_{l}\right)  )>0$. Then there exists $\widehat{\boldsymbol{w}}_{^{l+1}%
}=\left(  \widehat{w}_{k}^{l+1}\right)  _{k\in G_{l}}$ such that
\begin{equation}%
%TCIMACRO{\dsum \limits_{\boldsymbol{v}_{l}}}%
%BeginExpansion
{\displaystyle\sum\limits_{\boldsymbol{v}_{l}}}
%EndExpansion
\exp\left(
%TCIMACRO{\dsum \limits_{_{k\in G_{l}}}}%
%BeginExpansion
{\displaystyle\sum\limits_{_{k\in G_{l}}}}
%EndExpansion
\widehat{w}_{k}^{l+1}v_{k}^{l}\right)  \left(  \boldsymbol{P}_{l}\left(
\boldsymbol{v}_{l};\boldsymbol{\theta}_{l}\right)  -\boldsymbol{Q}\left(
\boldsymbol{v}_{l}\right)  \right)  <0. \label{Inequality_1}%
\end{equation}

\end{lemma}

\begin{proof}
Take $\widehat{\boldsymbol{v}}_{l}\neq\boldsymbol{0}$ such that
$\boldsymbol{Q}\left(  \widehat{\boldsymbol{v}}_{l}\right)  \in\left(
0,1\right)  $. Then for any $\boldsymbol{w}_{l}^{\prime}$, $\boldsymbol{b}%
_{l}^{\prime}$ given there exists $\boldsymbol{a}_{l}^{\prime}$ such that
\begin{equation}
\boldsymbol{P}_{l}\left(  \widehat{\boldsymbol{v}}_{l};\boldsymbol{\theta}%
_{l}^{\prime}\right)  <\boldsymbol{Q}\left(  \widehat{\boldsymbol{v}}%
_{l}\right)  , \label{Condition_3}%
\end{equation}
where $\boldsymbol{\theta}_{l}^{\prime}=\left(  \boldsymbol{w}_{l}^{\prime
},\boldsymbol{a}_{l}^{\prime},\boldsymbol{b}_{l}^{\prime}\right)  $. If such
$\widehat{\boldsymbol{v}}_{l}=\left(  \widehat{v}_{k}^{l}\right)  _{k\in
G_{l}}$ does not exist, then $\boldsymbol{Q}\left(  \boldsymbol{v}_{l}\right)
$ is concentrated in one point, i.e. $\boldsymbol{Q}\left(  \boldsymbol{v}%
_{0}\right)  =1$. In this case $KL(\boldsymbol{Q}\left(  \boldsymbol{v}%
_{l}\right)  \mid\boldsymbol{P}_{l}\left(  \boldsymbol{v}_{l}%
;\boldsymbol{\theta}_{l}\right)  )=0$. But this case is ruled out by the
hypothesis $KL(\boldsymbol{Q}\left(  \boldsymbol{v}_{l}\right)  \mid
\boldsymbol{P}_{l}\left(  \boldsymbol{v}_{l};\theta_{l}\right)  )>0$.

We now set
\[
\boldsymbol{1}=\underset{\#G_{l.}\text{-times}}{\underbrace{\left(
1,\ldots,1\right)  }\text{ }}\text{and \ }\widehat{\boldsymbol{w}}%
_{l+1}=\left(  \widehat{w}_{k}^{l+1}\right)  _{k\in G_{l}}=\alpha
(\widehat{\boldsymbol{v}}_{l}-\frac{1}{2}\boldsymbol{1}),
\]
where $\alpha$ is a positive number. Then, for $\boldsymbol{v}_{l}\neq
\widehat{\boldsymbol{v}}_{l}$,
\begin{equation}
\lim_{\alpha\rightarrow\infty}\frac{\exp\left(
%TCIMACRO{\dsum \limits_{_{k\in G_{l}}}}%
%BeginExpansion
{\displaystyle\sum\limits_{_{k\in G_{l}}}}
%EndExpansion
\widehat{w}_{k}^{l+1}v_{k}^{l}\right)  }{\exp\left(
%TCIMACRO{\dsum \limits_{_{k\in G_{l}}}}%
%BeginExpansion
{\displaystyle\sum\limits_{_{k\in G_{l}}}}
%EndExpansion
\widehat{w}_{k}^{l+1}\widehat{v}_{k}^{l}\right)  }=0. \label{Condition_4}%
\end{equation}
A detailed verification of this last inequality can be found in the
demonstration of Theorem 1 in \cite{Le roux et al 1}. Consequently,%
\begin{gather*}%
%TCIMACRO{\dsum \limits_{\boldsymbol{v}_{l}}}%
%BeginExpansion
{\displaystyle\sum\limits_{\boldsymbol{v}_{l}}}
%EndExpansion
\exp\left(
%TCIMACRO{\dsum \limits_{_{k\in G_{l}}}}%
%BeginExpansion
{\displaystyle\sum\limits_{_{k\in G_{l}}}}
%EndExpansion
\widehat{w}_{k}^{l+1}v_{k}^{l}\right)  \left(  \boldsymbol{P}_{l}\left(
\boldsymbol{v}_{l};\boldsymbol{\theta}_{l}\right)  -\boldsymbol{Q}\left(
\boldsymbol{v}_{l}\right)  \right)  =\\
\exp\left(
%TCIMACRO{\dsum \limits_{_{k\in G_{l}}}}%
%BeginExpansion
{\displaystyle\sum\limits_{_{k\in G_{l}}}}
%EndExpansion
\widehat{w}_{k}^{l+1}\widehat{v}_{k}^{l}\right)  \left\{  \boldsymbol{P}%
_{l}\left(  \widehat{\boldsymbol{v}}_{l};\boldsymbol{\theta}_{l}\right)
-\boldsymbol{Q}\left(  \widehat{\boldsymbol{v}}_{l}\right)  +%
%TCIMACRO{\dsum \limits_{\boldsymbol{v}_{l}\neq\widehat{\boldsymbol{v}}_{l}}}%
%BeginExpansion
{\displaystyle\sum\limits_{\boldsymbol{v}_{l}\neq\widehat{\boldsymbol{v}}_{l}%
}}
%EndExpansion
\frac{\exp\left(
%TCIMACRO{\dsum \limits_{_{k\in G_{l}}}}%
%BeginExpansion
{\displaystyle\sum\limits_{_{k\in G_{l}}}}
%EndExpansion
\widehat{w}_{k}^{l+1}v_{k}^{l}\right)  }{\exp\left(
%TCIMACRO{\dsum \limits_{_{k\in G_{l}}}}%
%BeginExpansion
{\displaystyle\sum\limits_{_{k\in G_{l}}}}
%EndExpansion
\widehat{w}_{k}^{l+1}\widehat{v}_{k}^{l}\right)  }\right\}  ,
\end{gather*}
and by using (\ref{Condition_4}),%
\begin{multline*}%
%TCIMACRO{\dsum \limits_{\boldsymbol{v}_{l}}}%
%BeginExpansion
{\displaystyle\sum\limits_{\boldsymbol{v}_{l}}}
%EndExpansion
\exp\left(
%TCIMACRO{\dsum \limits_{_{k\in G_{l}}}}%
%BeginExpansion
{\displaystyle\sum\limits_{_{k\in G_{l}}}}
%EndExpansion
\widehat{w}_{k}^{l+1}v_{k}^{l}\right)  \left(  \boldsymbol{P}_{l}\left(
\boldsymbol{v}_{l};\boldsymbol{\theta}_{l}\right)  -\boldsymbol{Q}\left(
\boldsymbol{v}_{l}\right)  \right) \\
\sim\exp\left(
%TCIMACRO{\dsum \limits_{_{k\in G_{l}}}}%
%BeginExpansion
{\displaystyle\sum\limits_{_{k\in G_{l}}}}
%EndExpansion
\widehat{w}_{k}^{l+1}\widehat{v}_{k}^{l}\right)  \left(  \boldsymbol{P}%
_{l}\left(  \widehat{\boldsymbol{v}}_{l};\boldsymbol{\theta}_{l}\right)
-\boldsymbol{Q}\left(  \widehat{\boldsymbol{v}}_{l}\right)  \right)
\end{multline*}
as $\alpha\rightarrow\infty$. Finally, by using (\ref{Condition_3}), there
exists $\alpha_{0}$ such that (\ref{Inequality_1}) holds true for
$\alpha>\alpha_{0}$.
\end{proof}

\begin{remark}
Given positive integers $l$, $l_{0}$, with $l\geq l_{0}$, we identify
$G_{l_{0}}$ with the \textit{subset} of $G_{l}$ consisting of integers having
the form $i_{0}+i_{1}p+\ldots+i_{l_{0}-1}p^{l_{0}-1}$, where the $i_{k}$s are
$p$-adic digits.
\end{remark}

\begin{theorem}
\label{Theorem_2} Let $\boldsymbol{Q}(\boldsymbol{v})$ be an arbitrary
probability distribution on $\left\{  0,1\right\}  ^{m}$. As discussed above,
we assume without loss of generality that $m=p^{l_{0}}$. We identify
$\boldsymbol{v}$ with $\boldsymbol{v}_{l}=\left(  v_{j}^{l}\right)  _{j\in
G_{l}}$, and $\boldsymbol{Q}(\boldsymbol{v}_{l})$ with a probability
distribution on the $\boldsymbol{v}_{l}$s. Let $DBN(p,l,\boldsymbol{\theta
}_{l})$ be a $p$-adic discrete DBN, with $l\geq l_{0}$, such that
$KL(\boldsymbol{Q}\left(  \boldsymbol{v}_{l}\right)  \mid\boldsymbol{P}%
_{l}\left(  \boldsymbol{v}_{l};\boldsymbol{\theta}_{l}\right)  )>0$. Then the
two following assertions hold true.

\noindent(i) There exists an $DBN(p,l+1,\boldsymbol{\theta}_{l},\boldsymbol{w}%
_{l+1},b_{j_{0}}^{l+1})$ constructed from $DBN(p,l,\boldsymbol{\theta}_{l})$
by adding one layer with marginal probability distribution $\boldsymbol{P}%
_{l+1}\left(  \boldsymbol{v}_{l};\boldsymbol{\theta}_{l},\boldsymbol{w}%
_{l+1},b_{j_{0}}^{l+1}\right)  $ satisfying%
\begin{equation}
KL(\boldsymbol{Q}\left(  \boldsymbol{v}_{l}\right)  \mid\boldsymbol{P}%
_{l+1}\left(  \boldsymbol{v}_{l};\boldsymbol{\theta}_{l},\boldsymbol{w}%
_{l+1},b_{j_{0}}^{l+1}\right)  )<KL(\boldsymbol{Q}\left(  \boldsymbol{v}%
_{l}\right)  \mid\boldsymbol{P}_{l}\left(  \boldsymbol{v}_{l}%
;\boldsymbol{\theta}_{l}\right)  ), \label{Part_2A}%
\end{equation}
for some $\boldsymbol{\theta}_{l},\boldsymbol{w}_{l+1},b_{j_{0}}^{l+1}$.

\noindent(ii) Given $\epsilon>0$ arbitrarily small, there exists an
\[
DBN(p,l+k,\boldsymbol{\theta}_{l},\boldsymbol{w}_{l+1},\ldots.,\boldsymbol{w}%
_{l+k},b_{j_{0}}^{l+1},\ldots,b_{j_{k-1}}^{l+k})
\]
with marginal probability distribution $\boldsymbol{P}_{l+k}\left(
\boldsymbol{v}_{l};\boldsymbol{\theta}_{l},\boldsymbol{w}_{l+1},\ldots
.,\boldsymbol{w}_{l+k},b_{j_{0}}^{l+1},\ldots,b_{j_{k-1}}^{l+k}\right)  $
satisfying%
\begin{equation}
KL(\boldsymbol{Q}\left(  \boldsymbol{v}_{l}\right)  \mid\boldsymbol{P}%
_{l+k}\left(  \boldsymbol{v}_{l};\boldsymbol{\theta}_{l},\boldsymbol{w}%
_{l+1},\ldots.,\boldsymbol{w}_{l+k},b_{j_{0}}^{l+1},\ldots,b_{j_{k-1}}%
^{l+k}\right)  )<\epsilon, \label{Part_2}%
\end{equation}
where $k$ is a positive integer depending on $\epsilon$, for some
\[
\boldsymbol{\theta}_{l},\boldsymbol{w}_{l+1},\ldots.,\boldsymbol{w}%
_{l+k},b_{j_{0}}^{l+1},\ldots,b_{j_{k-1}}^{l+k}.
\]

\end{theorem}

\begin{proof}
(i) We first compute $KL(\boldsymbol{Q}\left(  \boldsymbol{v}_{l}\right)
\mid\boldsymbol{P}_{l+1}\left(  \boldsymbol{v}_{l};\boldsymbol{\theta}%
_{l+1}\right)  )$, $\boldsymbol{\theta}_{l+1}=\left(  \boldsymbol{\theta}%
_{l},\boldsymbol{w}_{l+1},b_{j_{0}}^{l+1}\right)  $, using formulas
(\ref{Formula_11})-(\ref{Formula_10}):%
\begin{gather}
KL(\boldsymbol{Q}\left(  \boldsymbol{v}_{l}\right)  \mid\boldsymbol{P}%
_{l+1}\left(  \boldsymbol{v}_{l};\boldsymbol{\theta}_{l+1}\right)  )=%
%TCIMACRO{\dsum \limits_{\boldsymbol{v}_{l}}}%
%BeginExpansion
{\displaystyle\sum\limits_{\boldsymbol{v}_{l}}}
%EndExpansion
\boldsymbol{Q}\left(  \boldsymbol{v}_{l}\right)  \ln\boldsymbol{Q}\left(
\boldsymbol{v}_{l}\right)  -%
%TCIMACRO{\dsum \limits_{\boldsymbol{v}_{l}}}%
%BeginExpansion
{\displaystyle\sum\limits_{\boldsymbol{v}_{l}}}
%EndExpansion
\boldsymbol{Q}\left(  \boldsymbol{v}_{l}\right)  \ln\boldsymbol{P}%
_{l+1}\left(  \boldsymbol{v}_{l};\boldsymbol{\theta}_{l+1}\right)
=\nonumber\\
-H(\boldsymbol{Q})-%
%TCIMACRO{\dsum \limits_{\boldsymbol{v}_{l}}}%
%BeginExpansion
{\displaystyle\sum\limits_{\boldsymbol{v}_{l}}}
%EndExpansion
\boldsymbol{Q}\left(  \boldsymbol{v}_{l}\right)  \ln\frac{\left\{
1+\exp\left(
%TCIMACRO{\dsum \limits_{k\in G_{l}}}%
%BeginExpansion
{\displaystyle\sum\limits_{k\in G_{l}}}
%EndExpansion
w_{k}^{l+1}v_{k}^{l}+b_{\boldsymbol{j}_{0}}^{l+1}\right)  \right\}
%TCIMACRO{\dsum \limits_{\boldsymbol{h}_{l}}}%
%BeginExpansion
{\displaystyle\sum\limits_{\boldsymbol{h}_{l}}}
%EndExpansion
\exp\left(  -E_{l}\left(  \boldsymbol{v}_{l},\boldsymbol{h}_{l}%
;\boldsymbol{\theta}_{l}\right)  \right)  }{%
%TCIMACRO{\dsum \limits_{\widetilde{\boldsymbol{v}}_{l},\widetilde
%{\boldsymbol{h}}_{l}}}%
%BeginExpansion
{\displaystyle\sum\limits_{\widetilde{\boldsymbol{v}}_{l},\widetilde
{\boldsymbol{h}}_{l}}}
%EndExpansion
\left\{  1+\exp\left(
%TCIMACRO{\dsum \limits_{k\in G_{l}}}%
%BeginExpansion
{\displaystyle\sum\limits_{k\in G_{l}}}
%EndExpansion
w_{k}^{l+1}\widetilde{v}_{k}^{l}+b_{j_{0}}^{l+1}\right)  \right\}  \exp\left(
-E_{l}\left(  \widetilde{\boldsymbol{v}}_{l},\widetilde{\boldsymbol{h}}%
_{l};\boldsymbol{\theta}_{l}\right)  \right)  }\nonumber\\
=-H(\boldsymbol{Q})-%
%TCIMACRO{\dsum \limits_{\boldsymbol{v}_{l}}}%
%BeginExpansion
{\displaystyle\sum\limits_{\boldsymbol{v}_{l}}}
%EndExpansion
\boldsymbol{Q}\left(  \boldsymbol{v}_{l}\right)  \ln\left(  1+\exp\left(
%TCIMACRO{\dsum \limits_{k\in G_{l}}}%
%BeginExpansion
{\displaystyle\sum\limits_{k\in G_{l}}}
%EndExpansion
w_{k}^{l+1}v_{k}^{l}+b_{j_{0}}^{l+1}\right)  \right) \nonumber\\
-%
%TCIMACRO{\dsum \limits_{\boldsymbol{v}_{l}}}%
%BeginExpansion
{\displaystyle\sum\limits_{\boldsymbol{v}_{l}}}
%EndExpansion
\boldsymbol{Q}\left(  \boldsymbol{v}_{l}\right)  \ln\left(
%TCIMACRO{\dsum \limits_{\boldsymbol{h}_{l}}}%
%BeginExpansion
{\displaystyle\sum\limits_{\boldsymbol{h}_{l}}}
%EndExpansion
\exp\left(  -E_{l}\left(  \boldsymbol{v}_{l},\boldsymbol{h}_{l}%
;\boldsymbol{\theta}_{l}\right)  \right)  \right) \nonumber\\
+\left(
%TCIMACRO{\dsum \limits_{\boldsymbol{v}_{l}}}%
%BeginExpansion
{\displaystyle\sum\limits_{\boldsymbol{v}_{l}}}
%EndExpansion
\boldsymbol{Q}\left(  \boldsymbol{v}_{l}\right)  \right)  \ln\left(
%TCIMACRO{\dsum \limits_{\widetilde{\boldsymbol{v}}_{l},\widetilde
%{\boldsymbol{h}}_{l}}}%
%BeginExpansion
{\displaystyle\sum\limits_{\widetilde{\boldsymbol{v}}_{l},\widetilde
{\boldsymbol{h}}_{l}}}
%EndExpansion
\left\{  1+\exp\left(
%TCIMACRO{\dsum \limits_{\boldsymbol{k}\in G_{l}}}%
%BeginExpansion
{\displaystyle\sum\limits_{\boldsymbol{k}\in G_{l}}}
%EndExpansion
w_{k}^{l+1}\widetilde{v}_{k}^{l}+b_{j_{0}}^{l+1}\right)  \right\}  \exp\left(
-E_{l}\left(  \widetilde{\boldsymbol{v}}_{l},\widetilde{\boldsymbol{h}}%
_{l};\boldsymbol{\theta}_{l}\right)  \right)  \right) \nonumber\\
=:-H(\boldsymbol{Q})-KL_{1}(\boldsymbol{\theta}_{l})-KL_{2}(\boldsymbol{w}%
_{l+1},b_{j_{0}}^{l+1})+KL_{3}\left(  \boldsymbol{\theta}_{l},\boldsymbol{w}%
_{l+1},b_{j_{0}}^{l+1}\right)  . \label{Formula_ultima}%
\end{gather}
Given any $\boldsymbol{w}_{l+1}$, we may assume that $\exp\left(  \sum_{k\in
G_{l}}w_{k}^{l+1}v_{k}^{l}+b_{j_{0}}^{l+1}\right)  $ is very small for any
$\boldsymbol{v}$, by taking $-b_{j_{0}}^{l+1}$ sufficiently large, since
$v_{k}\in\left\{  0,1\right\}  $. Then by using $\ln\left(  1+x\right)
=x+o(x)$ as $x\rightarrow0$, we have%
\begin{equation}
\ln\left(  1+\exp\left(
%TCIMACRO{\dsum \limits_{k\in G_{l}}}%
%BeginExpansion
{\displaystyle\sum\limits_{k\in G_{l}}}
%EndExpansion
w_{k}^{l+1}v_{k}^{l}+b_{j_{0}}^{l+1}\right)  \right)  =\exp\left(
%TCIMACRO{\dsum \limits_{k\in G_{l}}}%
%BeginExpansion
{\displaystyle\sum\limits_{k\in G_{l}}}
%EndExpansion
w_{k}^{l+1}v_{k}^{l}+b_{j_{0}}^{l+1}\right)  +o\left(  \exp\left(  b_{j_{0}%
}^{l+1}\right)  \right)  \text{, } \label{A_key_Obs}%
\end{equation}
as $b_{j_{0}}^{l+1}\rightarrow-\infty$. Then, the term $KL_{1}%
(\boldsymbol{\theta}_{l})$ becomes%
\begin{equation}
KL_{1}(\boldsymbol{\theta}_{l})=%
%TCIMACRO{\dsum \limits_{\boldsymbol{v}_{l}}}%
%BeginExpansion
{\displaystyle\sum\limits_{\boldsymbol{v}_{l}}}
%EndExpansion
\boldsymbol{Q}\left(  \boldsymbol{v}_{l}\right)  \exp\left(
%TCIMACRO{\dsum \limits_{k\in G_{l}}}%
%BeginExpansion
{\displaystyle\sum\limits_{k\in G_{l}}}
%EndExpansion
w_{k}^{l+1}v_{k}^{l}+b_{j_{0}}^{l+1}\right)  +o\left(  \exp\left(
b_{\boldsymbol{j}_{0}}^{l+1}\right)  \right)  \text{ as }b_{j_{0}}%
^{l+1}\rightarrow-\infty, \label{Formula_ultima_1}%
\end{equation}
and the term $KL_{3}\left(  \boldsymbol{\theta}_{l},\boldsymbol{w}%
_{l+1},b_{j_{0}}^{l+1}\right)  $ becomes%
\begin{multline*}
KL_{3}\left(  \boldsymbol{\theta}_{l},\boldsymbol{w}_{l+1},b_{j_{0}}%
^{l+1}\right)  =\\
\ln\left(
%TCIMACRO{\dsum \limits_{\widetilde{\boldsymbol{v}}_{l},\widetilde
%{\boldsymbol{h}}_{l}}}%
%BeginExpansion
{\displaystyle\sum\limits_{\widetilde{\boldsymbol{v}}_{l},\widetilde
{\boldsymbol{h}}_{l}}}
%EndExpansion
\left\{  1+\exp\left(
%TCIMACRO{\dsum \limits_{k\in G_{l}}}%
%BeginExpansion
{\displaystyle\sum\limits_{k\in G_{l}}}
%EndExpansion
w_{k}^{l+1}\widetilde{v}_{k}^{l}+b_{j_{0}}^{l+1}\right)  \right\}  \exp\left(
-E_{l}\left(  \widetilde{\boldsymbol{v}}_{l},\widetilde{\boldsymbol{h}}%
_{l};\boldsymbol{\theta}_{l}\right)  \right)  \right) \\
=\ln\left(
%TCIMACRO{\dsum \limits_{\widetilde{\boldsymbol{v}}_{l},\widetilde
%{\boldsymbol{h}}_{l}}}%
%BeginExpansion
{\displaystyle\sum\limits_{\widetilde{\boldsymbol{v}}_{l},\widetilde
{\boldsymbol{h}}_{l}}}
%EndExpansion
\exp\left(  -E_{l}\left(  \widetilde{\boldsymbol{v}}_{l},\widetilde
{\boldsymbol{h}}_{l};\boldsymbol{\theta}_{l}\right)  \right)  \right)  +\\
\ln\left(  1+\frac{%
%TCIMACRO{\dsum \limits_{\widetilde{\boldsymbol{v}}_{l},\widetilde
%{\boldsymbol{h}}_{l}}}%
%BeginExpansion
{\displaystyle\sum\limits_{\widetilde{\boldsymbol{v}}_{l},\widetilde
{\boldsymbol{h}}_{l}}}
%EndExpansion
\exp\left(
%TCIMACRO{\dsum \limits_{k\in G_{l}}}%
%BeginExpansion
{\displaystyle\sum\limits_{k\in G_{l}}}
%EndExpansion
w_{k}^{l+1}\widetilde{v}_{k}^{l}+b_{j_{0}}^{l+1}\right)  \exp\left(
-E_{l}\left(  \widetilde{\boldsymbol{v}}_{l},\widetilde{\boldsymbol{h}}%
_{l};\boldsymbol{\theta}_{l}\right)  \right)  }{%
%TCIMACRO{\dsum \limits_{\widetilde{\boldsymbol{v}}_{l},\widetilde
%{\boldsymbol{h}}_{l}}}%
%BeginExpansion
{\displaystyle\sum\limits_{\widetilde{\boldsymbol{v}}_{l},\widetilde
{\boldsymbol{h}}_{l}}}
%EndExpansion
\exp\left(  -E_{l}\left(  \widetilde{\boldsymbol{v}}_{l},\widetilde
{\boldsymbol{h}}_{l};\boldsymbol{\theta}_{l}\right)  \right)  }\right)  .
\end{multline*}
Now, by using (\ref{A_key_Obs}), we have for $b_{j_{0}}^{l+1}\rightarrow
-\infty$ that
\begin{gather}
KL_{3}\left(  \boldsymbol{\theta}_{l},\boldsymbol{w}_{l+1},b_{j_{0}}%
^{l+1}\right)  =\ln\left(
%TCIMACRO{\dsum \limits_{\widetilde{\boldsymbol{v}}_{l},\widetilde
%{\boldsymbol{h}}_{l}}}%
%BeginExpansion
{\displaystyle\sum\limits_{\widetilde{\boldsymbol{v}}_{l},\widetilde
{\boldsymbol{h}}_{l}}}
%EndExpansion
\exp\left(  -E_{l}\left(  \widetilde{\boldsymbol{v}}_{l},\widetilde
{\boldsymbol{h}}_{l};\boldsymbol{\theta}_{l}\right)  \right)  \right)
\label{Formula_ultima_2}\\
+\frac{%
%TCIMACRO{\dsum \limits_{\widetilde{\boldsymbol{v}}_{l},\widetilde
%{\boldsymbol{h}}_{l}}}%
%BeginExpansion
{\displaystyle\sum\limits_{\widetilde{\boldsymbol{v}}_{l},\widetilde
{\boldsymbol{h}}_{l}}}
%EndExpansion
\exp\left(
%TCIMACRO{\dsum \limits_{k\in G_{l}}}%
%BeginExpansion
{\displaystyle\sum\limits_{k\in G_{l}}}
%EndExpansion
w_{k}^{l+1}\widetilde{v}_{k}^{l}+b_{j_{0}}^{l+1}\right)  \exp\left(
-E_{l}\left(  \widetilde{\boldsymbol{v}}_{l},\widetilde{\boldsymbol{h}}%
_{l};\boldsymbol{\theta}_{l}\right)  \right)  }{%
%TCIMACRO{\dsum \limits_{\widetilde{\boldsymbol{v}}_{l},\widetilde
%{\boldsymbol{h}}_{l}}}%
%BeginExpansion
{\displaystyle\sum\limits_{\widetilde{\boldsymbol{v}}_{l},\widetilde
{\boldsymbol{h}}_{l}}}
%EndExpansion
\exp\left(  -E_{l}\left(  \widetilde{\boldsymbol{v}}_{l},\widetilde
{\boldsymbol{h}}_{l};\boldsymbol{\theta}_{l}\right)  \right)  }+o\left(
\exp(b_{j_{0}}^{l+1})\right) \nonumber\\
=\ln\left(
%TCIMACRO{\dsum \limits_{\widetilde{\boldsymbol{v}}_{l},\widetilde
%{\boldsymbol{h}}_{l}}}%
%BeginExpansion
{\displaystyle\sum\limits_{\widetilde{\boldsymbol{v}}_{l},\widetilde
{\boldsymbol{h}}_{l}}}
%EndExpansion
\exp\left(  -E_{l}\left(  \widetilde{\boldsymbol{v}}_{l},\widetilde
{\boldsymbol{h}}_{l};\boldsymbol{\theta}_{l}\right)  \right)  \right)  +%
%TCIMACRO{\dsum \limits_{\widetilde{\boldsymbol{v}}_{l}}}%
%BeginExpansion
{\displaystyle\sum\limits_{\widetilde{\boldsymbol{v}}_{l}}}
%EndExpansion
\exp\left(
%TCIMACRO{\dsum \limits_{k\in G_{l}}}%
%BeginExpansion
{\displaystyle\sum\limits_{k\in G_{l}}}
%EndExpansion
w_{k}^{l+1}\widetilde{v}_{k}^{l}+b_{j_{0}}^{l+1}\right)  P_{l}\left(
\widetilde{\boldsymbol{v}}_{l};\boldsymbol{\theta}_{l}\right) \nonumber\\
+o\left(  \exp(b_{j_{0}}^{l+1})\right)  .\nonumber
\end{gather}
Finally, from formulas (\ref{Formula_ultima})-(\ref{Formula_ultima_2}), we
obtain that%
\begin{gather*}
KL(\boldsymbol{Q}\left(  \boldsymbol{v}_{l}\right)  \mid\boldsymbol{P}%
_{l+1}\left(  \boldsymbol{v}_{l};\boldsymbol{\theta}_{l+1}\right)
)-KL(\boldsymbol{Q}\left(  \boldsymbol{v}_{l}\right)  \mid\boldsymbol{P}%
_{l}\left(  \boldsymbol{v}_{l};\boldsymbol{\theta}_{l}\right)  )=\\%
%TCIMACRO{\dsum \limits_{\boldsymbol{v}_{l}}}%
%BeginExpansion
{\displaystyle\sum\limits_{\boldsymbol{v}_{l}}}
%EndExpansion
\exp\left(
%TCIMACRO{\dsum \limits_{k\in G_{l}}}%
%BeginExpansion
{\displaystyle\sum\limits_{k\in G_{l}}}
%EndExpansion
w_{k}^{l+1}v_{k}^{l}+b_{j_{0}}^{l+1}\right)  \left(  P_{l}\left(
\boldsymbol{v}_{l};\theta_{l}\right)  -\boldsymbol{Q}\left(  \boldsymbol{v}%
_{l}\right)  \right)  +o\left(  \exp(b_{j_{0}}^{l+1})\right)
\end{gather*}
as $o\left(  \exp(b_{j_{0}}^{l+1})\right)  \rightarrow-\infty$. By Applying
Lemma \ref{Lemma_3}, there exist $\widetilde{\boldsymbol{w}}_{l+1}%
$,\ $\widetilde{b}_{j_{0}}^{l+1}$ such that
\[
KL(\boldsymbol{Q}\left(  \boldsymbol{v}_{l}\right)  \mid\boldsymbol{P}%
_{l+1}\left(  \boldsymbol{v}_{l};\boldsymbol{\theta}_{l},\widetilde
{\boldsymbol{w}}_{l+1},\widetilde{b}_{j_{0}}^{l+1}\right)  )-KL(\boldsymbol{Q}%
\left(  \boldsymbol{v}_{l}\right)  \mid\boldsymbol{P}_{l}\left(
\boldsymbol{v}_{l};\boldsymbol{\theta}_{l}\right)  )<0.
\]

(ii) We proceed recursively. If $KL(\boldsymbol{Q}\left(  \boldsymbol{v}%
_{l}\right)  \mid\boldsymbol{P}_{l+1}\left(  \boldsymbol{v}_{l}%
;\boldsymbol{\theta}_{l},\boldsymbol{w}_{l+1},b_{j_{0}}^{l+1}\right)  )=0$,
for some $\boldsymbol{\theta}_{l},\boldsymbol{w}_{l+1},b_{j_{0}}^{l+1}$, then
the $DBN(p,l+1,\boldsymbol{\theta}_{l},\boldsymbol{w}_{l+1},b_{j_{0}}^{l+1})$
satisfies the condition required. Otherwise, by using the fact that the key
construction can be used in a recursive way, we use the part (i) to construct
a Boltzmann machine
\[
DBN(p,l+2,\boldsymbol{\theta}_{l+1},\boldsymbol{w}_{l+1},\boldsymbol{w}%
_{l+2},b_{j_{0}}^{l+1},b_{j_{1}}^{l+2}),
\]
which satisfies%
\begin{multline*}
KL(\boldsymbol{Q}\left(  \boldsymbol{v}_{l}\right)  \mid\boldsymbol{P}%
_{l+2}\left(  \boldsymbol{v}_{l};\boldsymbol{\theta}_{l+1},\boldsymbol{w}%
_{l+1},\boldsymbol{w}_{l+2},b_{j_{0}}^{l+1},b_{j_{1}}^{l+2}\right)  )\\
<KL(\boldsymbol{Q}\left(  \boldsymbol{v}_{l}\right)  \mid\boldsymbol{P}%
_{l+1}\left(  \boldsymbol{v}_{l};\boldsymbol{\theta}_{l},\boldsymbol{w}%
_{l+1},b_{j_{0}}^{l+1}\right)  )<KL(\boldsymbol{Q}\left(  \boldsymbol{v}%
_{l}\right)  \mid\boldsymbol{P}_{l}\left(  \boldsymbol{v}_{l}%
;\boldsymbol{\theta}_{l}\right)  ).
\end{multline*}
Therefore there exists $k(\epsilon)$ such that (\ref{Part_2}) holds true.
\end{proof}

\begin{theorem}
\label{Theorem_3}Let $\boldsymbol{Q}(\boldsymbol{v})$ be an arbitrary
probability distribution on $\left\{  0,1\right\}  ^{m}$. As discussed above,
we assume without loss of generality that $m=p^{l_{0}}$. We identify
$\boldsymbol{v}$ with $\boldsymbol{v}_{l}=\left(  v_{j}^{l}\right)  _{j\in
G_{l}}$, and $\boldsymbol{Q}(\boldsymbol{v}_{l})$ with a probability
distribution on the $\boldsymbol{v}_{l}$s. Then $\boldsymbol{Q}\left(
\boldsymbol{v}_{l}\right)  $ can be approximated arbitrarily well, in the
sense of the $KL$ divergence, by an $DBN(p,l_{0}+k,\boldsymbol{\theta}%
_{l_{0}+k})$, where $k$ is the number of input vectors whose probability in
not zero.
\end{theorem}

\begin{proof}
The argument is an adaptation of the one given in \cite{Le roux et al 1} for
Theorem 2. The key is observation is that adding a hidden unit in
\ \cite[Theorem 2]{Le roux et al 1} corresponds to add a level in our
construction. Furthermore, in \cite[Theorem 2]{Le roux et al 1} the marginal
distribution with an extra hidden unit $p_{w,c}(\boldsymbol{v})$ agrees with
\begin{multline*}
\boldsymbol{P}_{l+1}\left(  \boldsymbol{v}_{l};\boldsymbol{\theta}%
_{l},\boldsymbol{w}_{l+1},b_{j_{0}}^{l+1}\right)  =\\
\frac{\left\{  1+\exp\left(
%TCIMACRO{\dsum \limits_{k\in G_{l}}}%
%BeginExpansion
{\displaystyle\sum\limits_{k\in G_{l}}}
%EndExpansion
w_{k}^{l+1}v_{k}^{l}+b_{j_{0}}^{l+1}\right)  \right\}
%TCIMACRO{\dsum \limits_{\boldsymbol{h}_{l}}}%
%BeginExpansion
{\displaystyle\sum\limits_{\boldsymbol{h}_{l}}}
%EndExpansion
\exp\left(  -E_{l}\left(  \boldsymbol{v}_{l},\boldsymbol{h}_{l}%
;\boldsymbol{\theta}_{l}\right)  \right)  }{%
%TCIMACRO{\dsum \limits_{\boldsymbol{v}_{l},\boldsymbol{h}_{l}}}%
%BeginExpansion
{\displaystyle\sum\limits_{\boldsymbol{v}_{l},\boldsymbol{h}_{l}}}
%EndExpansion
\left\{  1+\exp\left(
%TCIMACRO{\dsum \limits_{k\in G_{l}}}%
%BeginExpansion
{\displaystyle\sum\limits_{k\in G_{l}}}
%EndExpansion
w_{k}^{l+1}v_{k}^{l}+b_{j_{0}}^{l+1}\right)  \right\}  \exp\left(
-E_{l}\left(  \boldsymbol{v}_{l},\boldsymbol{h}_{l};\boldsymbol{\theta}%
_{l}\right)  \right)  },
\end{multline*}
where $\boldsymbol{\theta}_{l}=\left(  \boldsymbol{w}_{l},\boldsymbol{a}%
_{l},\boldsymbol{b}_{l}\right)  $, up to the function $E_{l}\left(
\boldsymbol{v}_{l},\boldsymbol{h}_{l};\boldsymbol{\theta}_{l}\right)  $. Let
$\widetilde{\boldsymbol{v}}_{l}=\left(  \widetilde{v}_{k}^{l}\right)  _{k\in
G_{l}}$ be an arbitrary input vector and let $\widehat{\boldsymbol{w}}_{l+1}$
be the vector defined as the proof of Lemma \ref{Lemma_3}:
\[
\widehat{\boldsymbol{w}}_{l+1}=\left[  \widehat{w}_{k}^{l+1}\right]  _{_{k\in
G_{l}}}\text{, \ }\widehat{w}_{k}^{l+1}=\alpha(\widetilde{v}_{k}^{l}-\frac
{1}{2})\text{ for }k\in G_{l},
\]
where $\alpha$ is a positive number. We define $\widehat{b}_{j_{0}}%
^{l+1}=-\sum_{_{k\in G_{l}}}\widehat{w}_{k}^{l+1}\widetilde{v}_{k}^{l}%
+\lambda$, with $\lambda\in\mathbb{R}$. Then%
\[
\lim_{\alpha\rightarrow\infty}1+\exp\left(  \sum_{_{k\in G_{l}}}\widehat
{w}_{k}^{l+1}\widetilde{v}_{k}^{l}+\widehat{b}_{j_{0}}^{l+1}\right)  =\left\{
\begin{array}
[c]{lll}%
1 & \text{if} & \boldsymbol{v}_{l}\neq\widetilde{\boldsymbol{v}}_{l}\\
&  & \\
1+\exp\lambda & \text{if} & \boldsymbol{v}_{l}=\widetilde{\boldsymbol{v}}_{l},
\end{array}
\right.
\]
and by using formula (\ref{Formula_11}), we have%
\begin{equation}
\lim_{\alpha\rightarrow\infty}\boldsymbol{P}_{l+1}\left(  \boldsymbol{v}%
;\boldsymbol{\theta}_{l},\widehat{\boldsymbol{w}}_{l+1},\widehat{b}_{j_{0}%
}^{l+1}\right)  =\left\{
\begin{array}
[c]{ccc}%
\frac{\boldsymbol{P}_{l}\left(  \boldsymbol{v}_{l};\boldsymbol{\theta}%
_{l}\right)  }{1+\exp\left(  \lambda\right)  \boldsymbol{P}_{l}\left(
\widetilde{\boldsymbol{v}}_{l};\boldsymbol{\theta}_{l}\right)  } & \text{if} &
\boldsymbol{v}_{l}\neq\widetilde{\boldsymbol{v}}_{l}\\
&  & \\
\frac{\left(  1+\exp\left(  \lambda\right)  \right)  \boldsymbol{P}_{l}\left(
\widetilde{\boldsymbol{v}}_{l};\boldsymbol{\theta}_{l}\right)  }{1+\exp\left(
\lambda\right)  \boldsymbol{P}_{l}\left(  \widetilde{\boldsymbol{v}}%
_{l};\boldsymbol{\theta}_{l}\right)  } & \text{if} & \boldsymbol{v}%
_{l}=\widetilde{\boldsymbol{v}}_{l}.
\end{array}
\right.  \label{Key_Calculation_1}%
\end{equation}
By choosing a suitable value of $\lambda$, and by adding an extra level to an
$DBN(p,l,\boldsymbol{\theta}_{l})$, the probability of an arbitrary input
$\widetilde{\boldsymbol{v}}_{l}$ can be increased, while the probability of
any other input $\boldsymbol{v}_{l}\neq\widetilde{\boldsymbol{v}}_{l}$ can be
uniformly decreased by a multiplicative factor. Now, the required DBN can be
constructed recursively using the technique given in the proof of Theorem 2
in\ \cite{Le roux et al 1}. We index the input vectors as $\boldsymbol{u}_{i}$
where $i$ is an integer from $1$ to $2^{m}$, $m=p^{l_{0}}$, and sort them such
that%
\[
0=\boldsymbol{Q}(\boldsymbol{u}_{k+1})=\ldots=\boldsymbol{Q}(\boldsymbol{u}%
_{2^{m}})<\boldsymbol{Q}(\boldsymbol{u}_{1})\leq\boldsymbol{Q}(\boldsymbol{u}%
_{2})\leq\ldots\leq\boldsymbol{Q}(\boldsymbol{u}_{k}).
\]
We denote by
\[
\boldsymbol{P}_{l_{0}+r}\left(  \boldsymbol{v}_{l_{0}}\right)  =\boldsymbol{P}%
_{l_{0}+r}\left(  \boldsymbol{v}_{l_{0}};\boldsymbol{\theta}_{l_{0}%
+r-1},\boldsymbol{w}_{l_{0}+1},\ldots,\boldsymbol{w}_{l_{0}+r},b_{j_{0}%
}^{l_{0}+1},\ldots,b_{j_{r-1}}^{l_{0}+r}\right)  \text{,}%
\]
for $r=1,2,\ldots$, the marginal distribution of an RBM constructed from
\[
DBN\left(  p,l_{0},\boldsymbol{\theta}_{l_{0}},\boldsymbol{w}_{l_{0}%
+1},b_{j_{0}}^{l_{0}+1}\right)
\]
by using the the key construction $r$ times. The $\boldsymbol{P}_{l_{0}+r}$s
are defined inductively as follows. If $r=0$, we take $\boldsymbol{P}_{l_{0}%
}\left(  \boldsymbol{v}_{l_{0}}\right)  =2^{-m}$, $\boldsymbol{v}\in G_{l_{0}%
}$, is the uniform distribution. We now set%
\[
\widehat{\boldsymbol{w}}_{l_{0}+1}=\alpha\left(  \boldsymbol{u}_{1}-\frac
{1}{2}\right)  \text{ and }\widehat{b}_{j_{0}}^{l_{0}+1}=-\left\langle
\widehat{\boldsymbol{w}}_{l_{0}+1},\boldsymbol{v}_{1}\right\rangle
+\lambda_{1},
\]
where
\[
\left\langle \boldsymbol{w}_{_{l_{0}}},\boldsymbol{v}_{_{l_{0}}}\right\rangle
:=%
%TCIMACRO{\dsum \limits_{j\in G_{l_{0}}}}%
%BeginExpansion
{\displaystyle\sum\limits_{j\in G_{l_{0}}}}
%EndExpansion
w_{j}^{_{l_{0}}}v_{j}^{_{l_{0}}}.
\]
By (\ref{Key_Calculation_1}),%
\[
\lim_{\alpha\rightarrow\infty}\boldsymbol{P}_{l_{0}+1}\left(  \boldsymbol{v}%
;\boldsymbol{\theta}_{l_{0}},\widehat{\boldsymbol{w}}_{l_{0}+1},\widehat
{b}_{j_{0}}^{l_{0}+1}\right)  =\left\{
\begin{array}
[c]{ccc}%
\frac{\left(  1+\exp\left(  \lambda_{1}\right)  \right)  \boldsymbol{2}^{-m}%
}{1+\exp\left(  \lambda_{1}\right)  \boldsymbol{2}^{-m}} & \text{if} &
\boldsymbol{v}_{l_{0}}=\boldsymbol{u}_{1}\\
&  & \\
\frac{\boldsymbol{2}^{-m}}{1+\exp\left(  \lambda_{1}\right)  \boldsymbol{2}%
^{-m}} & \text{if} & \boldsymbol{v}_{l_{0}}=\boldsymbol{u}_{i}\text{, }%
i\geq2\text{.}%
\end{array}
\right.
\]
We now add an extra level to $DBN\left(  p,l_{0},\boldsymbol{\theta}_{l_{0}%
},\widehat{\boldsymbol{w}}_{l_{0}+1},\widehat{b}_{j_{0}}^{l_{0}+1}\right)  $
using the key construction. By choosing $\lambda_{2}$, $\boldsymbol{P}%
_{l_{0}+2}\left(  \boldsymbol{v}_{l_{0}}\right)  $ satisfies%
\[
\frac{\boldsymbol{P}_{l_{0}+2}\left(  \boldsymbol{u}_{2}\right)
}{\boldsymbol{P}_{l_{0}+2}\left(  \boldsymbol{u}_{1}\right)  }=\frac
{\boldsymbol{Q}\left(  \boldsymbol{u}_{2}\right)  }{\boldsymbol{Q}\left(
\boldsymbol{u}_{1}\right)  }.
\]
By using this construction recursively, one constructs a probability
distribution $\boldsymbol{P}_{l_{0}+k}\left(  \boldsymbol{v}_{l_{0}}\right)  $
satisfying%
\begin{align*}
\frac{\boldsymbol{P}_{l_{0}+k}\left(  \boldsymbol{u}_{k}\right)
}{\boldsymbol{P}_{l_{0}+k}\left(  \boldsymbol{u}_{k-1}\right)  }  &
=\frac{\boldsymbol{Q}\left(  \boldsymbol{u}_{k}\right)  }{\boldsymbol{Q}%
\left(  \boldsymbol{u}_{k-1}\right)  }\text{, \ldots,}\frac{\boldsymbol{P}%
_{l_{0}+2}\left(  \boldsymbol{u}_{2}\right)  }{\boldsymbol{P}_{l_{0}+2}\left(
\boldsymbol{u}_{1}\right)  }=\frac{\boldsymbol{Q}\left(  \boldsymbol{u}%
_{2}\right)  }{\boldsymbol{Q}\left(  \boldsymbol{u}_{1}\right)  },\\
\boldsymbol{P}_{l_{0}+k}\left(  \boldsymbol{u}_{k+1}\right)   &
=\ldots=\boldsymbol{P}_{l_{0}+k}\left(  \boldsymbol{u}_{2^{m}}\right)  .
\end{align*}
The solution of the above recursive system is given in the proof of Theorem 2
in \cite{Le roux et al 1}:%
\[
\boldsymbol{P}_{l_{0}+k}\left(  \boldsymbol{u}_{i}\right)  =\left\{
\begin{array}
[c]{lll}%
\frac{\boldsymbol{Q}\left(  \boldsymbol{u}_{1}\right)  }{1+\exp\left(
\lambda_{1}\right)  +\left(  2^{m}-k\right)  \boldsymbol{Q}\left(
\boldsymbol{v}_{1}\right)  } & \text{if} & i>k\\
&  & \\
\boldsymbol{Q}\left(  \boldsymbol{u}_{i}\right)  \frac{1+\exp\left(
\lambda_{1}\right)  }{1+\exp\left(  \lambda_{1}\right)  +\left(
2^{m}-k\right)  \boldsymbol{Q}\left(  \boldsymbol{v}_{1}\right)  } & \text{if}
& i\leq k.
\end{array}
\right.
\]
Finally,%
\[
KL(\boldsymbol{Q}\mid\boldsymbol{P}_{l_{0}+k})=%
%TCIMACRO{\dsum \limits_{i}}%
%BeginExpansion
{\displaystyle\sum\limits_{i}}
%EndExpansion
\boldsymbol{Q}\left(  \boldsymbol{u}_{i}\right)  \frac{\left(  2^{m}-k\right)
\boldsymbol{Q}\left(  \boldsymbol{u}_{i}\right)  }{1+\exp\left(  \lambda
_{1}\right)  }+o(\exp\left(  -\lambda_{1}\right)  )\rightarrow0
\]
as $\lambda_{1}\rightarrow\infty.$
\end{proof}

\section{\label{Section4}Discussion}

\subsection{Euclidean QFTs and NNs}

The literature about the connections between QFTs with NNs and brain activity
is extremely large. In this section we compare our results and our approach
with some recent works. We also propose several new open problems.

In \cite{Halverson et al}, the authors propose a correspondence between QFTs
and NNs. Many modern network architectures admits a Gaussian limit as the
number of neurons per layer tends to infinity. In the limit, these networks
can be described by Gaussian processes which naturally correspond to
non-interacting field theories. Moving away from the asymptotic limit yields
to non-Gaussian processes \ which are connected with interacting fields
theories. In our approach we work exclusively with interacting field theories:
a continuous version and a discrete version. See Table 1.
\begin{align*}
&
\begin{tabular}
[t]{|c|c|c|}\hline
$%
\begin{array}
[c]{c}%
p\text{-adic discrete DBN}\\
\text{with }l\text{ layers}%
\end{array}
$ & $\underleftarrow{\text{Discretization}}$ & $%
\begin{array}
[c]{c}%
p\text{-adic continuous DBN}\\
\text{with infinitely many layers}%
\end{array}
$\\\hline
$\Updownarrow$ &  & $\Updownarrow$\\\hline
$%
\begin{array}
[c]{c}%
\text{Discrete SFT determined}\\
\text{ by }E_{l}(\boldsymbol{v}_{l},\boldsymbol{h}_{l})
\end{array}
$ & $\underleftarrow{\text{Discretization}}$ & $%
\begin{array}
[c]{c}%
\text{SFT determined by}\\
\text{ }E(\boldsymbol{v},\boldsymbol{h})
\end{array}
$\\\hline
$\Updownarrow$ &  & $\Updownarrow$\\\hline
\multicolumn{1}{|l|}{$%
\begin{array}
[c]{c}%
\text{Probability measure}\\
\boldsymbol{P}_{l}(\boldsymbol{v}_{l},\boldsymbol{h}_{l})\text{ }d^{\#G_{l}%
}\boldsymbol{v}\text{ }d^{\#G_{l}}\boldsymbol{h}\text{ }\\
\text{on }\mathcal{D}^{l}(\mathbb{Z}_{p})\times\mathcal{D}^{l}(\mathbb{Z}_{p})
\end{array}
$} & $\underrightarrow{\text{Limit}}$ & \multicolumn{1}{|l|}{$%
\begin{array}
[c]{c}%
\text{Probability measure}\\
\frac{e^{-E(\boldsymbol{v},\boldsymbol{h})}}{Z^{\text{phys}}}d\boldsymbol{v}%
d\boldsymbol{h}\text{ }\\
\text{on }\mathcal{D}(\mathbb{Z}_{p})\times\mathcal{D}(\mathbb{Z}_{p})
\end{array}
$}\\\hline
\multicolumn{3}{|c|}{}\\\hline
$DBN(p,l,\boldsymbol{\theta}_{l})$ & $\underrightarrow{\text{Scaling, }m>l}$ &
$DBN(p,m,\boldsymbol{\theta}_{m})$\\\hline
\multicolumn{3}{|c|}{}\\\hline
\end{tabular}
\\
&  \text{{\small Table 1.\ The table provides a basic dictionary bewteen
Euclidean QFTs and NNs. The}}\\
&  \text{{\small last line in the table means that }}{\small DBN(p,m,\theta
}_{m}{\small )}\text{ {\small is a larger and computationally more}}\\
&  \text{{\small powerful version of }}{\small DBN(p,l,\theta}_{l}%
{\small )}\text{.}%
\end{align*}

A rigorous mathematical study of the following problem plays a central role in
the understanding the neural networks using statistical field theory:

\begin{problem}
Determine all the energy functionals $E(\boldsymbol{v},\boldsymbol{h})$ such
that
\begin{equation}
\frac{e^{-E(\boldsymbol{v},\boldsymbol{h})}}{Z^{\text{phys}}}d\boldsymbol{v}%
d\boldsymbol{h}\overset{\text{def}}{=}\lim_{l\rightarrow\infty}\boldsymbol{P}%
_{l}(\boldsymbol{v}_{l},\boldsymbol{h}_{l})\text{ }d^{\#G_{l}}\boldsymbol{v}%
\text{ }d^{\#G_{l}}\boldsymbol{h} \label{limit}%
\end{equation}
exists in some sense.
\end{problem}

In \cite{Zunifa-RMP-2022}, the author establishes, in a rigorous mathematical
way, the existence of $\phi^{4}$-interacting Euclidean quantum field theories
on a $p$-adic spacetime for which the limit (\ref{limit}) exists. In a
forthcoming publication we plan to expand the results given in
\cite{Zunifa-RMP-2022} to case of two fields and find the energy functionals
$E(\boldsymbol{v},\boldsymbol{h})$ for which the limit (\ref{limit}) exists.
The mentioned limit suggest that the correlation functions of the continuous
STF can be very well-approximated by the correlation functions of the
corresponding discrete SFT.

In \cite{Batchits et al}, authors study a generalization of the RBMs
associated with energy functionals of type:
\begin{gather}
S\left(  \boldsymbol{v},\boldsymbol{h};\boldsymbol{\theta}\right)  =-%
%TCIMACRO{\dsum \limits_{j\in\mathcal{G}}}%
%BeginExpansion
{\displaystyle\sum\limits_{j\in\mathcal{G}}}
%EndExpansion%
%TCIMACRO{\dsum \limits_{k\in\mathcal{G}}}%
%BeginExpansion
{\displaystyle\sum\limits_{k\in\mathcal{G}}}
%EndExpansion
w_{k,j}v_{j}h_{j}+%
%TCIMACRO{\dsum \limits_{j\in\mathcal{G}}}%
%BeginExpansion
{\displaystyle\sum\limits_{j\in\mathcal{G}}}
%EndExpansion
a_{j}v_{j}+%
%TCIMACRO{\dsum \limits_{j\in\mathcal{G}}}%
%BeginExpansion
{\displaystyle\sum\limits_{j\in\mathcal{G}}}
%EndExpansion
b_{j}h_{j}\label{Type_I}\\
+%
%TCIMACRO{\dsum \limits_{j\in\mathcal{G}}}%
%BeginExpansion
{\displaystyle\sum\limits_{j\in\mathcal{G}}}
%EndExpansion
c_{j}v_{j}^{2}+%
%TCIMACRO{\dsum \limits_{j\in\mathcal{G}}}%
%BeginExpansion
{\displaystyle\sum\limits_{j\in\mathcal{G}}}
%EndExpansion
d_{j}h_{j}^{2}+%
%TCIMACRO{\dsum \limits_{j\in\mathcal{G}}}%
%BeginExpansion
{\displaystyle\sum\limits_{j\in\mathcal{G}}}
%EndExpansion
f_{j}v_{j}^{4}+%
%TCIMACRO{\dsum \limits_{j\in\mathcal{G}}}%
%BeginExpansion
{\displaystyle\sum\limits_{j\in\mathcal{G}}}
%EndExpansion
g_{j}h_{j}^{4}\text{,}\nonumber
\end{gather}
where $\mathcal{G}$ is a square lattice. These generalizations are not DBNs
due to the topology of $\mathcal{G}$. Also, the authors assume that and
$w_{k,j}\neq0\Leftrightarrow$ $i$ and $j$ are connected by one edge. This
condition implies that the functional $S\left(  \boldsymbol{v},\boldsymbol{h}%
;\boldsymbol{\theta}\right)  $ is local. Our action $E_{l}\left(
\boldsymbol{v},\boldsymbol{h};\boldsymbol{\theta}\right)  $ is non local,
which means that (in general) $w_{i-j}\neq0$ for any $i$, $j\in G_{l}$. In
\cite{Batchits et al}, the authors also discussed the implementation of
several learning algorithms. In forthcoming article, we will discuss the
implementation of $p$-adic discrete DBNs based on energy functionals of type
(\ref{Type_I}) with $\mathcal{G}=G_{l}$.

In \cite{Grosvenor}, \cite{Moritz Dahmen} a completely different approach for
the correspondence between Euclidean QFTs and NNs is presented. Starting with
a stochastic differential equation, which plays the role of a master equation
for the neural network, the authors construct an action and a path integral,
which provides the QFT attached to the network. The non-Archimedean
counterpart of this construction is an open problem. Before considering this
problem, it is necessary to study non-Archimedean versions of stochastic
recurrent neural networks (SRNNs), see, e.g., \cite{Lim}-\cite{Lim et al} and
the references therein. Based on \cite{Lim}-\cite{Lim et al},
\cite{Zambrano-Zuniga}, \cite{Zuniga-JFAA-2015}, we propose the following
non-Archimedean version of the SRNNs:

\begin{problem}
Let $t\in\left[  0,T\right]  $ and let $\boldsymbol{v}\in C\left(  \left[
0,T\right]  ,L^{2}\left(  \mathbb{Q}_{p},dx\right)  \right)  $ be a
deterministic input signal. A $p$-adic temporal and spatially continuous SRNN
is described \ the following state-space model:%
\begin{align}
\frac{d\boldsymbol{h}\left(  x,t\right)  }{dt}  &  =-a\boldsymbol{h}\left(
x,t\right)  +\alpha\left(
%TCIMACRO{\dint \limits_{\mathbb{Q}_{p}}}%
%BeginExpansion
{\displaystyle\int\limits_{\mathbb{Q}_{p}}}
%EndExpansion
A(x,y)\boldsymbol{h}\left(  y,t\right)  dy+%
%TCIMACRO{\dint \limits_{\mathbb{Q}_{p}}}%
%BeginExpansion
{\displaystyle\int\limits_{\mathbb{Q}_{p}}}
%EndExpansion
B(x,y)\boldsymbol{v}\left(  y,t\right)  dy+B(x)\right) \nonumber\\
&  +\beta\left(  \boldsymbol{h}\left(  x,t\right)  ,\boldsymbol{v}\left(
x,t\right)  \right)  \overset{\cdot}{W}\left(  x,t\right)  \label{Equation_1}%
\\
& \nonumber\\
\boldsymbol{y}\left(  x,t\right)   &  =\sigma\left(  \boldsymbol{h}\left(
x,t\right)  \right)  . \label{Equation_2}%
\end{align}
Where (\ref{Equation_1}) is a stochastic equation for the hidden state
$\boldsymbol{h}\in C\left(  \left[  0,T\right]  ,L^{2}\left(  \mathbb{Q}%
_{p},dx\right)  \right)  $, $a>0$, $\alpha,\sigma:\mathbb{R}\rightarrow
\mathbb{R}$ are Lipschitz continuous and bounded functions, $A$, $B$ $\in
L^{1}\left(  \mathbb{Q}_{p}^{2},d^{2}x\right)  $, $\beta:\mathbb{R}%
\rightarrow\mathbb{R}$, and $\overset{\cdot}{W}\left(  x,t\right)  $ is the
formal notation for a Gaussian random perturbation defined on some probability
space, and (\ref{Equation_2}) defines the response of the network,
$\boldsymbol{y}\in C\left(  \left[  0,T\right]  ,L^{2}\left(  \mathbb{Q}%
_{p},dx\right)  \right)  $. A relevant problem is to study the response of\ SRNNs.
\end{problem}

\subsection{The non-Archimedean counterpart of the Buice-Cowan theory}

Buice and Cowan formulated a theory of fluctuating activity of cortical
networks in the language of stochastic fields, see, e.g., \cite{Buice2},
\cite{Buice-cowan-chow}, see also \cite{Coombes et al}. A relevant observation
is that the discrete master equation of the spike model can be formulated on
$G_{l}$. Consider a network of $N=p^{l}$ neurons. The configuration of each
neuron is given by the number of \ effective spikes $n_{i}$ that neuron $i\in
G_{l}$ has emitted. There is a weight function $w_{i,j}$ describing the
relative innervation of neuron $i$ by neuron $j$. We assume that $w_{i,j}$ is
a function of $\left\vert i-j\right\vert _{p}$, which is exactly the
hypothesis used by Buice and Cowan. The probability per unit of time \ that a
neuron \ will emit another spike is given by $f\left(  \sum_{j\in G_{l}%
}w_{i,j}n_{j}\right)  $. The state of the system is given by the probability
distribution $P_{\boldsymbol{n}}(t)$, which is the probability that the
network is in configuration \ $\boldsymbol{n}=\left(  n_{i}\right)  _{i\in
G_{l}}$ at the time $t$. The master equation of the network has the form%
\begin{gather}
\frac{\partial}{\partial t}P_{n_{i}}(t)=-\alpha n_{i}P_{n_{i}}(t)+\alpha
\left(  n_{i}+1\right)  P_{n_{i}+1}(t)\nonumber\\
-f\left(
%TCIMACRO{\dsum \limits_{\substack{j\neq i\\j,i\in G_{l}}}}%
%BeginExpansion
{\displaystyle\sum\limits_{\substack{j\neq i\\j,i\in G_{l}}}}
%EndExpansion
w_{i,j}n_{j}\right)  \left(  P_{n_{i}}(t)-P_{n_{i}+1}(t)\right)  ,
\label{Master equation}%
\end{gather}
for $i\in G_{l}$. Here $\alpha$ represents a decay rate, which is used to
account the fact that spikes are effective only for a time interval of
approximately $\frac{1}{\alpha}$. By using Doi techniques, see, e.g.,
\cite{Buice2}, \cite{Chow}, \cite{Moritz Dahmen}, involving
creation-annihilation operator formalism, the dynamics of the network can be
described \ by a vacuum ket $\left\vert 0\right\rangle $ using a pair
creation-annihilation operators at each site: $\left[  \Phi_{i},\Phi_{j}%
^{\dag}\right]  =\delta_{i,j}$. The state of the system is described by
\[
\left\vert \phi\left(  t\right)  \right\rangle =%
%TCIMACRO{\dsum \limits_{\boldsymbol{n}}}%
%BeginExpansion
{\displaystyle\sum\limits_{\boldsymbol{n}}}
%EndExpansion
P_{\boldsymbol{n}}(t)%
%TCIMACRO{\dprod \limits_{i\in G_{l}}}%
%BeginExpansion
{\displaystyle\prod\limits_{i\in G_{l}}}
%EndExpansion
\Phi_{i}^{n_{i}\dagger}\left\vert 0\right\rangle ,
\]
where the summation is taken over all configurations $\boldsymbol{n}$. In this
operator formalism the master equation takes the form%
\[
\frac{\partial}{\partial t}\left\vert \phi\left(  t\right)  \right\rangle
=-\widehat{H}\left\vert \phi\left(  t\right)  \right\rangle ,
\]
where%
\[
\widehat{H}=%
%TCIMACRO{\dsum \limits_{i\in G_{l}}}%
%BeginExpansion
{\displaystyle\sum\limits_{i\in G_{l}}}
%EndExpansion
\alpha\Phi_{i}^{\dag}\Phi_{i}-%
%TCIMACRO{\dsum \limits_{i\in G_{l}}}%
%BeginExpansion
{\displaystyle\sum\limits_{i\in G_{l}}}
%EndExpansion
\Phi_{i}^{\dag}f\left(
%TCIMACRO{\dsum \limits_{j\in G_{l}}}%
%BeginExpansion
{\displaystyle\sum\limits_{j\in G_{l}}}
%EndExpansion
w_{i,j}\left[  \Phi_{j}^{\dag}\Phi_{j}+\Phi_{j}\right]  \right)  .
\]
Now by using the work of Peliti, see e.g. \cite{Buice2}, \cite{Chow},
\cite{Moritz Dahmen}, the Hamiltonian $\widehat{H}$ can be studied using a
path integral. The corresponding action takes the following form (in the
\ continuous time limit):%
\begin{gather}
S\left[  \phi_{i}\left(  t\right)  ,\widetilde{\phi}_{i}\left(  t\right)
\right]  =\label{Action_1}\\%
%TCIMACRO{\dint \nolimits_{0}^{t}}%
%BeginExpansion
{\displaystyle\int\nolimits_{0}^{t}}
%EndExpansion
dt\left\{
%TCIMACRO{\dsum \limits_{i\in G_{l}}}%
%BeginExpansion
{\displaystyle\sum\limits_{i\in G_{l}}}
%EndExpansion
\widetilde{\phi}_{i}\partial_{t}\phi_{i}+\alpha\phi_{i}\widetilde{\phi}%
_{i}-\widetilde{\phi}_{i}f\left(
%TCIMACRO{\dsum \limits_{j\in G_{l}}}%
%BeginExpansion
{\displaystyle\sum\limits_{j\in G_{l}}}
%EndExpansion
w_{i,j}\left[  \widetilde{\phi}_{j}\phi_{j}+\phi_{j}\right]  \right)
\right\}  ,\nonumber
\end{gather}
where $dt$ denotes the Lebesgue measure of the real line. To obtain the
continuous limit in the spatial variable $i\rightarrow x$ in (\ref{Action_1}),
Buice and Cowan assumed that the sum in (\ref{Action_1}) runs over a square
lattice and thus by a limit process a Riemann-type integral is obtained. In
our case, this sum gives rise to and integral with respect to the Haar measure
of $\mathbb{Q}_{p}$:%
\[
S\left[  \phi\left(  x,t\right)  ,\widetilde{\phi}\left(  x,t\right)  \right]
=%
%TCIMACRO{\dint \nolimits_{0}^{t}}%
%BeginExpansion
{\displaystyle\int\nolimits_{0}^{t}}
%EndExpansion
dt%
%TCIMACRO{\dint \nolimits_{\mathbb{Q}_{p}}}%
%BeginExpansion
{\displaystyle\int\nolimits_{\mathbb{Q}_{p}}}
%EndExpansion
dx\left\{
%TCIMACRO{\dsum \limits_{i\in G_{l}}}%
%BeginExpansion
{\displaystyle\sum\limits_{i\in G_{l}}}
%EndExpansion
\widetilde{\phi}\partial_{t}\phi+\alpha\phi_{i}\widetilde{\phi}_{i}%
-\widetilde{\phi}_{i}f\left(  w\ast\left[  \widetilde{\phi}\phi+\phi\right]
\right)  \right\}  .
\]

\begin{problem}
To develop a non-Archimedean counterpart of the Buice-Cowan theory.
\end{problem}

\subsection{Final comments}

The connections between statistical mechanics and deep learning has been
studied intensively in the last ten years, see, e.g., \cite{Advani et al},
\cite{Bahri et al}, \cite{Carleo et al}, \cite{Engel et al}, \cite{Katsnelson
et al}, \cite{Metha et al}, \cite{Mezard et al}, \cite{Sohl-Dickstein et al},
\cite{Vanchurin}. To the best of our knowledge the theoretical results
presented here about the $p$-adic DBNs are new. Here we should mention that
$p$-adic neural networks have been considered before in
\cite{Albeverio-Khrennikov-Tirozzi}, \cite{Krennikov-Nilson}%
-\cite{Krennikov-tirozzi}, \cite{Zambrano-Zuniga}, but these computational
models are completely different to the ones considered here. We finally
mention that Khrennikov et al. have developed hierarchical models (based in
$p$-adic numbers) for brain activity and EEG analysis with applications in the
diagnosis of mental diseases see e.g. \cite{Khrennikov2}-\cite{Khrennikov4},
\cite{Shor et al}.

\section{\label{Section5}Conclusions}

In this work we initiated the study of the correspondence between $p$-adic
SFTs and NNs. An important advantage of the $p$-adic SFTS over the classical
ones is that discretization process that produces discrete SFTs can be carried
out in rigorous mathematical way in many relevant cases, for instant in the
$\phi^{4}$-theories \cite{Zunifa-RMP-2022}. A $p$-adic discrete SFT
corresponds to an energy functional defined on a tree $G_{l}$, this functional
defines a NN whose neurons are organized hierarchically in a tree-like
structure. This type of networks are a new particular class of the deep belief
networks introduced by Hinton et al. \cite{Hinton et al}-\cite{Honglak et al}.
A DBN is constructed by stacking several RBMs, the goal of this construction
is to get a network where the neurons are organized hierarchically in large
tree-like structure (a deep learning architecture). A classical RBM correspond
naturally to a certain spin glass, we argue a DBN should correspond to an
ultrametric spin glass.

A $p$-adic continuous DBN is a SFT, in this case the neurons correspond to the
points of the $p$-adic unit ball $\mathbb{Z}_{p}$, and thus, the neurons are
organized in an infinite rooted tree. A discrete version of this theory
corresponds to a $p$-adic discrete DBN. Intuitively, the discrete version is
obtained by cutting the tree $\mathbb{Z}_{p}$ at level $l\geq1$. It is
expected that in the limit $l$ tends infinite the discrete theories approach
to the continuous ones. This behavior is radically different to the one
presented in \cite{Halverson et al}. In this work, in the limit when the
number of neurons tend to infinity the network corresponds to a
non-interacting QFT, while in the finite case corresponds a interacting QFT.
Here, in both cases we have interacting QFTs. The $p$-adic discrete DBNs are
universal approximators, Theorems \ref{Theorem_2}, \ref{Theorem_3}. To
establish this result we adapt the techniques developed by Le Roux and Bengio
in \cite{Le roux et al 1}.

It is important emphasize that the correspondence between SFTs and NNs takes
different forms depending on the architecture of the networks. The case in
which the architecture is embedded in a stochastic differential equation or in
a discrete mater equation has been studied intensively lately, see e.g.
\cite{Batchits et al}, \cite{Grosvenor}, \cite{Halverson et al}, see also
\cite{Buice2}, \cite{Coombes et al}. The $p$-adic counterpart of this
correspondence is an open problem. It is widely accepted that the brain
activity is organized hierarchically. Here we pointed out that the Buice-Cowan
theory of fluctuating activity of cortical networks has a $p$-adic
counterpart, where the neurons are organized in an infinite tree-like
structure. In our view, the fully development of the $p$-adic counterpart of
the Buice-Cowan theory is a relevant matter.

\end{document}